\documentclass[aos]{imsart}

\RequirePackage{amsthm,amsmath,amsfonts,amssymb}
\RequirePackage[colorlinks,citecolor=blue,urlcolor=blue]{hyperref}
\RequirePackage[authoryear]{natbib}
\RequirePackage{graphicx}

\usepackage{mathtools}
\usepackage{tikz} 
\usepackage{amsthm}

\usepackage{hyperref}
\usepackage[capitalise]{cleveref}
\usepackage{multicol}
\usepackage{algorithm}
\usepackage[noend]{algorithmic}

\startlocaldefs
\numberwithin{equation}{section}
\theoremstyle{plain}

\newtheorem{theorem}{Theorem}[section]
\newtheorem{lemma}[theorem]{Lemma}
\newtheorem*{problem}{Problem}
\newtheorem{assumption}{Assumption}
\newtheorem{corollary}[theorem]{Corollary}
\theoremstyle{definition}
\newtheorem{definition}[theorem]{Definition}
\newtheorem{example}{Example}[section]
\newtheorem{remark}{Remark}[section]

\endlocaldefs

\newcommand{\LL}{\mathcal{L}}

\newcommand{\GG}{\mathcal{G}}

\newcommand{\MM}{\mathcal{M}}

\newcommand{\PP}{\mathcal{P}}

\newcommand{\R}{\mathbb{R}}

\usetikzlibrary{shapes,arrows,arrows.meta,positioning,calc,chains,shapes.symbols}
\usetikzlibrary{automata}

\newcommand{\indep}{\perp\negmedspace\!\!\perp}

\newcommand{\Ed}{E_{\rightarrow{}}}
\newcommand{\Eb}{E_{\leftrightarrow{}}}

\DeclareMathOperator{\rank}{rank}

\DeclareMathOperator{\HSIC}{HSIC}
\DeclareMathOperator{\trace}{tr}

\DeclareMathOperator*{\argmin}{arg\,min}
\DeclareMathOperator{\maxflow}{max-flow}
\DeclareMathOperator{\Sym}{Sym}

\def\newop#1{\expandafter\def\csname #1\endcsname{\mathop{\rm
#1}\nolimits}}

\newop{Inv}
\newop{conv}
\newop{pa}
\newop{Pa}
\newop{de}
\newop{nd}
\newop{GL}
\newop{O}
\newop{ch}
\newop{CS}
\newop{diag}
\newop{Var}
\newop{top}
\newop{sib}
\newop{Sib}
\newop{sp}
\newop{an}

\begin{document}

\begin{frontmatter}
\title{Parameter identification in linear non-Gaussian causal models under general confounding}
\runtitle{Identification in linear non-Gaussian causal models}

\begin{aug}
\author[A]{\fnms{Daniele}~\snm{Tramontano}\ead[label=e1]{daniele.tramontano@tum.de}},
\author[A]{\fnms{Mathias}~\snm{Drton}\ead[label=e2]{mathias.drton@tum.de}}
\and
\author[A]{\fnms{Jalal}~\snm{Etesami}\ead[label=e3]{j.etesami@tum.de}}
\address[A]{Technical University of Munich; School of Computation, Information and Technology, Germany\printead[presep={,\ }]{e1,e2,e3}}
\end{aug}

\begin{abstract}
Linear non-Gaussian causal models postulate that each random variable is a linear function of parent variables and non-Gaussian exogenous error terms. We study identification of the linear coefficients when such models contain latent variables. Our focus is on the commonly studied acyclic setting, where each model corresponds to a directed acyclic graph (DAG).
For this case, prior literature has demonstrated that connections to overcomplete independent component analysis yield effective criteria to decide parameter identifiability in latent variable models. However, this connection is based on the assumption that the observed variables linearly depend on the latent variables. Departing from this assumption, we treat models that allow for arbitrary non-linear latent confounding. Our main result is a graphical criterion that is necessary and sufficient for deciding the generic identifiability of direct causal effects. Moreover, we provide an algorithmic implementation of the criterion with a run time that is polynomial in the number of observed variables. Finally, we report on estimation heuristics based on the identification result and explore a generalization to models with feedback loops.
\end{abstract}


\end{frontmatter}

\section{Introduction}
Graphical models, directed or undirected, are ubiquitous in modern statistical practice for modeling multivariate distributions; see, e.g., \cite{lauritzen:1996, handbook}. In particular, structural equation models (SEMs) associated with directed \emph{acyclic} graphs (DAGs) provide a concise and effective way of stating the additional assumptions necessary to identify the causal parameters of interest \citep{pearl:2009,sprites:caus:pred}. As argued in \citet{pearl:2013}, understanding a causal phenomenon for linear SEMs is often a necessary step towards a generalized understanding of the same in a nonparametric framework.

The specific framework we focus on in this paper are linear non-Gaussian models that allow for general, possibly non-linear, confounding. Each such model is naturally associated to an acyclic directed mixed graph (ADMG); compare \cite{richardson:2023,zhao:2025} and references therein. Let $X=(X_v)_{v\in V}$ be a collection of observed random variables, and let $\GG=(V,\Ed,\Eb)$ be an ADMG whose vertices index the random variables $X_v$ and which contains two types of edge sets $\Ed,\Eb\subseteq~V\times~V\setminus\{(v,v): v\in V\}$. The edges in $\Ed$ are directed, and we depict them by $u\xrightarrow[]{} v$. Those in $\Eb$ are bidirected and depicted by $u\xleftrightarrow{}v$. The model encoded by $\GG$ posits that
\begin{equation}
\label{eq:admg:model}
X_v = \sum_{w\,:\,w\to v\in \Ed} \lambda_{wv}X_w + \varepsilon_v,\quad v\in V,
\end{equation}
where possible latent confounding is subsumed in the error variables $\varepsilon_v$, which are allowed to be dependent in accordance with the bidirected edges in $\Eb$.  In particular, if $v\xleftrightarrow{} w\in \Eb$, then the two errors $\varepsilon_v$ and $\varepsilon_w$ may be (arbitrarily) dependent.

\begin{example}
\label{ex:iv}
To illustrate the above definition, consider the graph from \cref{fig:IV}. It specifies an instrumental variable model for the joint distribution of three observed variables. The equations from~\eqref{eq:admg:model} take the form:
\begin{align}
X_1 &= \varepsilon_1, & 
X_2 &= \lambda_{12}X_1+\varepsilon_2, &
X_3 &= \lambda_{23}X_2+\varepsilon_3,
\end{align}
where $\varepsilon_1$ is independent of $(\varepsilon_2,\varepsilon_3)$ but $\varepsilon_2$ and $\varepsilon_3$ may be dependent.
\end{example}

\begin{figure}[t]
    \centering
    \scalebox{0.9}{
        \begin{tikzpicture}[scale = 0.3,->,>=triangle 45,shorten >=1pt,
            auto, thick,
            main node/.style={rectangle,draw,font=\sffamily\bfseries}] 
          
            \node[main node,rounded corners] (1) {$X_1:$ Tax Rate}; 
            \node[main node,rounded corners] (2) [right=1.5cm of 1] {$X_2:$ Mom's Smoking}; 
            \node[main node,rounded corners] (3) [right=1.5cm of 2] {$X_3:$ Baby's Weight};
            
            \path[every node/.style={font=\sffamily\small}] 
            (1) edge node[below] {$\lambda_{12}$} (2)
            (2) edge node[below] {$\lambda_{23}$} (3);

            \draw[dashed][<->] (2) to[bend left=40]node[main node,rounded corners,fill=none,above=0.2cm] {\color{black} Confounders} (3);
          \end{tikzpicture}
      }
    \caption{Instrumental variable graph based on \cite{evans:1999}.}
    \label{fig:IV}
\end{figure}
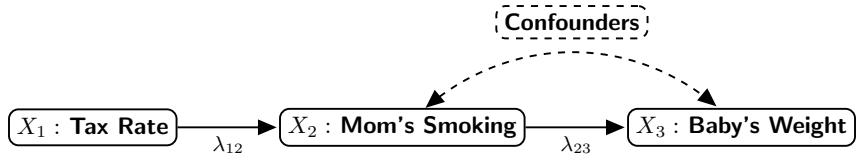

In this paper, we treat the problem of deciding which of the direct causal effects $\lambda_{wv}$ in~\eqref{eq:admg:model} can be identified from the joint distribution of the vector of observed variables $X$.  Our main results give a complete graphical characterization in terms of the ADMG $\GG$, an efficient algorithm to check the resulting identifiability criterion, and a simple practical estimation method that is based on empirical measures of dependences among estimates of the errors $(\varepsilon_v)_{v\in V}$.  We should highlight that our characterization targets \emph{generic} identifiability, which is the notion most suitable for problems such as the instrumental variable model from \cref{ex:iv}.  There, the key coefficient of interest $\lambda_{23}$ is identified as a ratio of covariances, $\text{Cov}[X_1, X_3]/\text{Cov}[X_1, X_2]$, but only if the denominator is nonzero which requires the genericity constraint that $\lambda_{12}\not=0$.  As part of our results, we also develop a framework for making genericity conditions for the infinite-dimensional set of non-Gaussian error distributions, which we justify via cumulants truncated at arbitrary order.

\subsection{Related Work}
For fully nonparametric SEMs, the ID algorithm \citep{shpitser:2006, kivva2022revisiting, shpitser:2023, kivva2023identifiability} is sound and complete for determining global identifiability in a given ADMG.  In contrast to the generic setting treated in this paper, global identifiability requires identifiability under every single distribution in the model. It follows from the work of \citet{drton:2011} that the  graphical criterion underpinning the ID algorithm also applies to \emph{global} identifiability within linear Gaussian models.  However, the graphical prerequisites for achieving global identification are frequently overly restrictive. For example, any ADMG containing a \emph{bow} (a pair of nodes, $u, v\in V$, such that $u\to v, u\xleftrightarrow{}v\in\GG$) would fail to meet the criteria for \emph{global} identifiability, thus overlooking significant scenarios such as the IV model illustrated in \cref{fig:IV}. Consequently, when studying linear SEMs, researchers have shifted their focus towards \emph{generic} identifiability results, for which much progress has been made recently, but for which a complete characterization is still lacking; see, e.g., \citet{kumor:2020, foygel:2022,sturma:2025}.

Non-Gaussianity of the error term has been extensively employed to achieve identifiability of the graphical structure of causal models; see \citet{shimizu:2022} for a recent account. In contrast, its application to causal effect identification has received little attention \citep{salehkaleybar:2020,kivva:2023,shuai:2023, tramontano:2024:icml,tramontano:2025}. All the works just mentioned explicitly model the confounding as linear. This approach, on the one hand, allows one to draw on the vast literature on overcomplete independent component analysis (OICA) in order to obtain stronger identifiability results \citep{eriksson:2004}.  On the other hand, however, it restricts the possible confounding structures. Moreover, as opposed to ICA in the fully observed case, overcomplete ICA is not separable \citep[Thm.~4]{eriksson:2004}, implying that the only algorithms for solving OICA that come with theoretical guarantees require making parametric distributional assumptions and using ad-hoc EM-type algorithms to solve the optimization problem; see, e.g., \cite{lewicki:2000}.

The work most similar to ours in terms of distributional assumptions is that of \citet{wang:2023}, which focuses on structural identifiability. \citet{wang:2023} note that bow-free acyclic graphs are identifiable from observational data and provide an estimation algorithm for such graphs. \citet{liu:2021} extended the algorithm to learn graphs with multi-directed edges.
\subsection{Organization of the Paper}

The rest of the paper is organized as follows. Section \ref{subsec:notation} contains standard graphical model notation used in the rest of the paper. In \cref{sec:lingam:mixed}, we formally define the identifiability problem we study. Section \ref{sec:id:result} contains the main results of our work; we provide a necessary and sufficient graphical condition for generic identifiability in the model under study. In \cref{sec:cert:ident}, we prove that our criterion can be certified in polynomial time in the size of the graph. Section \ref{sec:gen:ass} contains a detailed analysis of the genericity assumption.
In \cref{sec:cycles}, we provide partial results about the identification for cyclic models. In \cref{sec:comp}, we note that when the identification criterion is met, the parameters can be estimated as the solution to a suitable optimization problem, and we present a simulation study to assess the performance of the estimation method. In \cref{sec:conc}, we draw final conclusions and suggest future research directions. \citet{tramontano:2024:supp} contains further preliminary material and details of the proofs.

\subsection{Notation}
\label{subsec:notation}

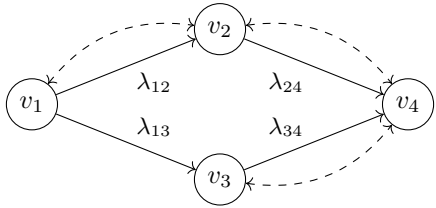
\begin{figure}[t]
    \centering
    \begin{tikzpicture}
            \node[draw, circle, inner sep=3pt] (A) at (-6,0) {$v_1$};
            \node[draw, circle, inner sep=3pt] (B) at (-3.5,1) {$v_2$};
            \node[draw, circle, inner sep=3pt] (C) at (-3.5,-1) {$v_3$};
            \node[draw, circle, inner sep=3pt] (D) at (-1,0) {$v_4$};
            
            \draw[->] (A) -- (B) node[pos = 0.7, below = 0.1] {\small $\lambda_{12}$};
            \draw[->] (A) -- (C) node[pos = 0.7, above = 0.1] {\small $\lambda_{13}$};
            \draw[->] (B) -- (D) node[pos = 0.3, below = 0.1] {\small $\lambda_{24}$};
            \draw[->] (C) -- (D) node[pos = 0.3, above = 0.1] {\small $\lambda_{34}$};

            \draw[dashed][<->] (A) to[bend left=30] (B);
            \draw[dashed][<->] (B) to[bend left=30] (D);
            \draw[dashed][<->] (C) to[bend right=30] (D);
        \end{tikzpicture}
        
    \caption{An acyclic directed mixed graphs (ADMG) with 4 nodes.}
    \label{fig:4:nodes:ex}
\end{figure}

A mixed graph is a triple $\GG=(V,\Ed,\Eb)$, where $\Ed,\Eb\subset~V\times~V\setminus\{(v,v): v\in V\}$.  We depict the pairs in $\Ed$ by $u\xrightarrow[]{} v$ and the ones in $\Eb$ by $u\xleftrightarrow{}v$; we refer to them as directed and bidirected edges, respectively.  For an example, see \cref{fig:4:nodes:ex}.

Let $u,v\in V$ be two vertices in $\GG$.
A \emph{directed} (respectively \emph{bidirected}) path from $u$ to $v$ is a sequence of nodes $\pi=(u=v_0,\dots, v_k=v)$ such that $(v_i, v_{i+1})\in\Ed$ (respectively $\Eb$) for all $i=0,\dots,k-1$.  This includes the case $k=0$, where $u=v$ and the path has no edges; we call such a path trivial.
We denote by $\mathcal{P}(u,v)$ the set of all \emph{directed} paths from $u$ to $v$.

A directed cycle is a \emph{non-trivial} directed path from a node $u$ to itself. The graph $\GG$ is acyclic if it contains no directed cycles; we refer to this class of graphs as acyclic directed mixed graphs (ADMG).  \cref{fig:4:nodes:ex} shows one instance. If the graph is acyclic, we can define a causal order on the nodes of $\GG$, that is, a total order $\le$ on $V$ such that $u\le v$ whenever $\PP(u,v)\not=\emptyset$. When considering parameter matrices associated to $\GG$, we will typically fix a causal order $\le$ and assume that the vertices in $V$ are enumerated as $v_1,\dots,v_p$ with $i\le j$ with $v_i\le v_j$; compare \cref{fig:4:nodes:ex}.

We will consider the following genealogical relations commonly used to indicate relationships between the vertices of an ADMG (parents, ancestors, children, descendants, siblings):
\begin{equation*}
    \begin{aligned}
        \pa(v) &:= \{u\in V\::\: u\to v\in\GG\}, \qquad&&\an(v) := \{u\in V\::\: \PP(u,v)\neq\emptyset\in\GG\}, \\
        \ch(v) &:= \{u\in V\::\: v\to u\in\GG\}, \qquad&&\de(v) := \{u\in V\::\: \PP(v,u)\neq\emptyset\in\GG\},\\
        \sib(v) &:= \{u\in V\::\: u\xleftrightarrow{} v\in\GG\},\qquad&& \Sib(v) := \sib(v)\cup \{v\}.
    \end{aligned}
\end{equation*} 
Note that $v\in\an(v)$ and $v\in\de(v)$, via trivial paths.  For a subset of vertices $U\subseteq V$, we define $\pa(U)=\cup_{u\in U}\pa(u)$ and make the analogous convention for the other relations.

Let $U=\{u_1,\dots,u_n\}$ and $W=\{w_1,\dots,w_n\}$ be two subsets of $V$ that have the same cardinality $n$, and for which we have fixed an ordering of their elements. Let $S_n$ be the symmetric group on $[n]=\{1,\dots,n\}$.
We say that $\Pi=(\pi_1,\dots,\pi_n)$ is a system of directed paths between $U$ and $W$, if there exists a permutation $\sigma_\Pi\in S_n$ such that $\pi_k\in\PP(u_k,w_{\sigma_\Pi(k)})$ for every $k\in[n]$. 
    We denote the set of all such systems by $\PP(U, W)$. A system $\Pi\in\PP(U,W)$ is called \emph{non-intersecting} if $\pi_k\cap \pi_l=\emptyset$ for $k\neq l$. The set of all non-intersecting systems in $\PP(U, W)$ is denoted by $\tilde{\PP}(U, W)$; see  \cref{fig:non-intersecting} for an example. 

    \begin{figure}[t]
        \centering
        (a)
                \begin{tikzpicture}
           \node[] (1) at (0,0.5) {$u_1$};
           \node[] (2) at (0,-0.5)   {$u_2$};
           \node[] (3) at (1.5,0.5)   {$w_1$};
           \node[] (4) at (1.5,-0.5)  {$w_2$};
           
           \draw[->] (1) to (3);
           \draw[->] (2) to (4);
        \end{tikzpicture}
            \hspace{1cm}
            (b)
        \begin{tikzpicture}
           \node[] (0) at (0,0) {$c$};
           \node[] (1) at (-1,0.5) {$u_1$};
           \node[] (2) at (-1,-0.5)   {$u_2$};
           \node[] (3) at (1,0.5)   {$w_1$};
           \node[] (4) at (1,-0.5)  {$w_2$};
           
           \draw[->] (1) to (0);
           \draw[->] (2) to (0);
           \draw[->] (0) to (3);
           \draw[->] (0) to (4);
        \end{tikzpicture}
        \caption{(a) A non-intersecting systems of directed paths from $\{u_1,u_2\}$ to $\{w_1,w_2\}$.  (b) Two intersecting directed paths with node $c$ in their intersection.}\label{fig:non-intersecting}
    \end{figure}
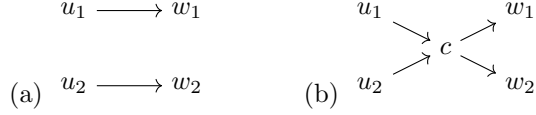

When connecting a graph $\GG$ to a statistical model, we will introduce a matrix of parameters whose entries act as weights on the directed edges. We will write $\R^{\GG_D}$ for the set of $p\times p$ real matrices $\Lambda = (\lambda_{uv})$ such that $\lambda_{uv}=0$ if $u\to v\notin \Ed$.
When $\GG$ is acyclic---as we assume throughout this work, the matrix $I-\Lambda$ is invertible for all $\Lambda\in\R^{\GG_D}$; here, $I$ denotes the identity matrix.  Indeed, when the nodes of $\GG$ are ordered according to a causal order, $(I-\Lambda)^T$ is lower triangular with all ones on the diagonal and $\det(I-\Lambda)=1$. We define $B_{\Lambda}:= (I-\Lambda)^{-T}$.  Later in Section \ref{sec:cycles}, we briefly discuss the identifiability problem in cyclic graphs, where the invertibility of $I-\Lambda$ becomes a modeling assumption.

Finally, let $U$ and $W$ be subsets of the row and column sets of a matrix $A$, respectively. We denote the submatrix containing only the rows in $U$ and the columns in $W$ as $A_{U,W}$.

\section{Linear Mixed Graph Models and Identifiability}
\label{sec:lingam:mixed}

Let $\GG=(V,\Ed,\Eb)$ be an ADMG, and let $\GG_D=(V,\Ed)$ and $\GG_B=(V,\Eb)$ be its directed and bidirected subgraphs, respectively. We say that a subset $C\subseteq V$ is \emph{connected} in the bidirected part $\GG_{B}$ if every pair of vertices $u,v\in C$ is joined by a \emph{bidirected} path in $\GG_{B}$, where every
vertex on the path is in $C$. On a fixed probability space, let $\varepsilon = (\varepsilon_1,\dots,\varepsilon_p)$ be a random vector taking values in $\R^p$ and satisfying the \emph{connected set Markov property} with respect to $\GG_{B}$ \citep{richardson:2003,drton:richardson:2008}, that is,
\[
\varepsilon_C\indep\varepsilon_{V\setminus\Sib(C)} \ \text{for all} \ \emptyset\neq C\subset V, \ C \text{ connected in } \GG_B.
\]
We denote the set of all such random vectors by $\mathcal{M}(\GG_B)$.  Note that the connected set Markov property implies but is generally stronger than requiring that $\varepsilon_u\indep\varepsilon_v$ for $u\xleftrightarrow{}v\notin\GG_B$.
\begin{definition}
\label{def:linearSEM}
The linear structural equation model $\mathcal{M}(\GG)$ corresponding to a mixed graph $\GG$ is the set of all $p$-variate real random vectors $X$ (on our fixed probability space) that solve the equation system
\begin{align*}
        X=\Lambda^T\cdot X+\varepsilon\iff X=(I-\Lambda)^{-T}\cdot\varepsilon=B_{\Lambda}\cdot\varepsilon,
\end{align*}
for a choice of $\Lambda\in\R^{\GG_D}$ and $\varepsilon\in\mathcal{M}(\GG_B)$.  The model $\mathcal{M}(\GG)$ is thus parametrized by the map
    \begin{equation*}
    \label{eq:phi:map}
        \begin{aligned}
            \Phi_{\GG}: \R^{\GG_D}\times\mathcal{M}(\GG_B)&\xrightarrow{}\MM(\GG)\\
             (\Lambda,\varepsilon) &\mapsto (I-\Lambda)^{-T}\varepsilon.
        \end{aligned}
    \end{equation*}
\end{definition}

\begin{example}
\label{ex:matrices}
    Let $\GG$ be the ADMG from \cref{fig:4:nodes:ex}. 
        The set $\mathcal{\MM}(\GG_B)$ contains all the random vectors $\varepsilon = (\varepsilon_1,\varepsilon_2,\varepsilon_3,\varepsilon_4)$ such that $(\varepsilon_1,\varepsilon_2)\indep \varepsilon_{3}$ and $\varepsilon_{1}\indep (\varepsilon_{3},\varepsilon_{4})$. 
        The space $\R^{\GG_D}$ is comprised of all matrices of the shape:
    $$
       \Lambda= 
       \setlength{\arraycolsep}{6pt}
       \begin{bmatrix}
            0 & \lambda_{12} & \lambda_{13} & 0 \\[3pt]
            0 & 0 & 0 & \lambda_{24}\\[3pt]
            0 & 0 & 0 & \lambda_{34}\\[3pt]
            0 & 0 & 0 & 0
        \end{bmatrix}.
    $$
    Accordingly, we have 
    $$
      B_{\Lambda}:=(I-\Lambda)^{-T}=  
      \setlength{\arraycolsep}{6pt}
      \begin{bmatrix}
            1 & 0 & 0 & 0 \\[3pt]
            \lambda_{12} & 1 & 0 & 0\\[3pt]
            \lambda_{13} & 0 & 1 & 0\\[3pt]
            \lambda_{12}\lambda_{24} + \lambda_{13}\lambda_{34} & \lambda_{24} & \lambda_{34} & 1
        \end{bmatrix}.
    $$
\end{example}

In this paper, we are concerned with parameter identifiability.  In other words, we ask under which conditions on $\GG$, the distribution of $X\in\mathcal{M}(\GG)$ uniquely determines entries of the coefficient matrix $\Lambda\in\R^{\GG_D}$ in the representation $X=(I-\Lambda)^{-T}\cdot\varepsilon$.
While we will not emphasize this in the sequel, the unique determination of all entries of $\Lambda$ also entails unique recovery of the distribution of $\varepsilon=(I-\Lambda)^TX$.  

As noted earlier, our interest is in a generic notion of identifiability, so we ask:

\begin{problem}
\label{prob:def}
    Under which graphical conditions on $\GG$ is a set of entries $\lambda_{uv}$ of the parameter matrix $\Lambda$ generically identifiable?
\end{problem}

To detail the problem, we make the following definition and then firm up the involved notion of genericity.

\begin{definition}
\label{def:id}
   We define the fiber of an element $X\in\MM(\GG)$ with respect to $\Phi_\GG$ as 
    \begin{equation}
    \label{eq:fibers}
        \Phi_\GG^{-1}(X):=\{(\Lambda, \varepsilon) \in \R^{\GG_D}\times\mathcal{M}(\GG_B) \;:\; \Phi_\GG(\Lambda, \varepsilon) \stackrel{d}{=} X\},
    \end{equation}
    where $\stackrel{d}{=}$ denotes equality in distribution.
    Let $\mathrm{Proj}_{\R^\GG}(\Phi_\GG^{-1}(X))$ be the projection of the set onto $\R^{\GG_D}$.  
    A parameter given by a function $f(\Lambda)$ is \emph{generically identifiable} if any \emph{generic} choice of $(\Lambda, \varepsilon)$ yields a random vector $X=(I-\Lambda)^{-T}\varepsilon$ for which  it holds that
    \[
    f(\Tilde{\Lambda})=f(\Lambda) \ \text{for all} \  \tilde{\Lambda}\in\mathrm{Proj}_{\R^\GG}(\Phi^{-1}_{\GG}(X)).
    \]
\end{definition}

Requiring genericity of $\Lambda$ will mean that we exclude a fixed Lebesgue null set of $\mathbb{R}^{\GG_D}$. For instance, in the instrumental variable (IV) example depicted in \cref{fig:IV}, the unknown coefficients are $(\lambda_{12},\lambda_{23})$ and $\R^{\GG_D}\equiv\R^2$. The coefficient $\lambda_{23}$ is identifiable outside the null set given by $\lambda_{12}=0$, i.e., we exclude the case that the instrument ($X_1$) does not affect the exposure ($X_2$).

While excluding null sets of a finite-dimensional space is a standard approach in related literature \citep[\S9]{drton:2018}, speaking of a generic choice of $\varepsilon$ requires clarification as Definition~\ref{def:linearSEM} is nonparametric with respect to the distribution of the errors $\varepsilon$.  Indeed, our genericity concept has a very specific meaning, namely, that the distribution of $\varepsilon$ satisfies the following assumption.

\begin{assumption}
    \label{ass:error}
    Let $\varepsilon\in\MM(\GG_B)$. For every two vectors $a_1,a_2\in\mathbb{R}^V$, it holds that\\ $a_1^T\varepsilon\indep a_2^T\varepsilon$ implies that $a_{1u}\cdot a_{2v}=0$ whenever $u=v$ or $u\xleftrightarrow{} v\in\Eb$.
\end{assumption}

Our genericity assumption is natural in view of the Darmois-Skitovich theorem.  Indeed, this theorem amounts to exactly the statement that if $\GG_B$ is the empty graph (i.e., has no edges), then \cref{ass:error} holds for every random vector that has at most one normally distributed coordinate.  To further justify our assumption, we present in \cref{sec:gen:ass} a detailed study of two different classes of submodels for which we show that indeed only a lower-dimensional set of distributions is excluded by our assumption.  One class of submodels is built by assuming the existence of moments up to an arbitrary but fixed order.  The other class is built by assuming linearity of confounding.   

For the remainder of this work, whenever we use the term \emph{generic}, it is implied that the result holds for any matrix $\Lambda$ outside of a fixed Lebesgue measure zero subset of $\mathbb{R}^{\mathcal{G}_D}$ and for any $\varepsilon \in \mathcal{M}(\mathcal{G}_B)$ that satisfies \cref{ass:error}.

\section{Necessary and Sufficient Conditions for Generic Identifiability of Direct Causal Effects}
\label{sec:id:result}
Let  $\GG=(V,\Ed,\Eb)$ be an ADMG, and let $X$ be a random vector in the model $\MM(\GG)$, i.e., 
\[
X=\Phi_\GG(\Lambda, \varepsilon) \ \text{for} \ (\Lambda, \varepsilon)\in\R^{\GG_D}\times\mathcal{M}(\GG_B).
\]
Suppose $X$ can be generated using another pair $(\tilde{\Lambda}, \tilde{\varepsilon})\in\Phi_\GG^{-1}(X)$. From the definition of the fiber in \cref{eq:fibers}, one can see that  
\begin{equation}
\label{eq:A:matrix}
    \tilde{\varepsilon}\stackrel{d}{=}(I-\tilde{\Lambda})^{T}(I-\Lambda)^{-T}\varepsilon=(I-\tilde{\Lambda})^{T}B_\Lambda\varepsilon=: A\varepsilon.
\end{equation}
The next result shows that the entries of matrix $A=(I-\tilde{\Lambda})^{T}B_\Lambda$ can be fully specified as a function of both $\Tilde{\Lambda}$ and $B_{\Lambda}$ through the ancestral relations among the nodes of $\GG$. 
\begin{lemma}
\label{lem:A:matrix}
The entries of matrix $A$ defined in \cref{eq:A:matrix} can be written as
\begin{equation}
\label{eq:a:matrix}
    a_{vu}=b_{vu}-\sum_{w\in\\\pa(v)\cap\de(u)}\tilde{\lambda}_{wv}b_{wu}.
\end{equation}
In particular, we have $a_{vu}=0$ if $v\notin\de(u)$, and $a_{uu}=1$ for every $u\in V$.
\end{lemma}
\begin{proof}
    Writing the product of matrices explicitly, we get 
$$
    a_{vu}=b_{vu}-\sum_{w\in V}\tilde{\lambda}_{wv}b_{wu}.
$$
From the definition of $\R^{\GG_D}$, we know that $\tilde{\lambda}_{wv}=0$ if $v\notin\ch(w)$, while it holds that $b_{wu}=0$ if $w\notin\de(u)$, from which the claim follows.  To see that $b_{wu}=0$ if $w\notin\de(u)$, note that $B_{\Lambda}$ is a path matrix for $\GG_D$ as we detail in \citet[Lemma 1.2]{tramontano:2024:supp}.
\end{proof}

\begin{definition} 
The set of removable ancestors of a node $v\in V$ is defined as
\[
R_v:=\{u\in\an(v)\::\: \Sib(u)\setminus\Sib(v)\neq\emptyset\} = \Sib(V\setminus\Sib(v))\cap \an(v).
\]
Clearly, $v\notin R_v$.
\end{definition}

\begin{example}
\label{ex:fig2} 
Consider the graph in \cref{fig:4:nodes:ex}. In this graph, the only strict ancestor of $v_2$ is $v_1$, which has only $v_2$ as its sibling. Hence, $R_{v_2}=\Sib(v_1)\setminus\Sib(v_2)=\{v_1,v_2\}\setminus\{v_1,v_2,v_4\}=\emptyset$. 
On the other hand, $R_{v_4} = \{v_1,v_2\}$ because 
$v_1$ belongs to both $\Sib(v_2)\setminus \Sib(v_4)$ and $\Sib(v_1)\setminus \Sib(v_4)$.
\end{example}

Using the concept of removable ancestors, the next result introduces a linear system of equations whose solution space fully characterizes the parameter matrices in $\mathrm{Proj}_{\R^\GG}(\Phi^{-1}_\GG(X))$.

\begin{lemma}
\label{lem:main}
Let $X=\Phi_\GG(\Lambda,\varepsilon)$ for a generic choice of parameters $(\Lambda,\varepsilon)\in\R^{\GG_D}\times\mathcal{M}(\GG_B)$. A matrix $\Tilde{\Lambda}\in\R^{\GG_D}$ belongs to $\mathrm{Proj}_{\R^\GG}(\Phi_{\GG}^{-1}(X))$ if and only if it is a solution to the following linear system of equations:
\begin{equation}
\label{eq:lambda:system}
    \underbrace{[(B_\Lambda)_{\pa(v),R_v}]^T}_{(B_\Lambda)^v}\cdot \Tilde{\lambda}_{\pa(v),v} =[(B_\Lambda)_{v,R_v}]^T,\quad \forall v\in V.
\end{equation}
\end{lemma}
\begin{proof}
     We start by showing the direct implication. \cref{eq:A:matrix} shows that for every $v_0\in V\setminus\{v\}$ we can write $\tilde{\varepsilon}_{\{v,v_0\}}=A_{\{v,v_0\},V}\cdot\varepsilon$, where $A_{\{v,v_0\},V}$ denotes the rows of $A$ corresponding to nodes $\{v,v_0\}$.  Since $\tilde{\varepsilon}\in\mathcal{M}(\GG_B)$, it holds that $\tilde{\varepsilon}_v\indep \tilde{\varepsilon}_{u_0}$ for every $u_0\notin\sib(v)$ and, thus,  \cref{ass:error} implies
         \begin{align}
            \label{eq:A:sys:1}
             a_{v v_0}\cdot a_{u_0 v_0} &= 0, \qquad&\forall\, v_0\in V,\\
             \label{eq:A:sys:2}
             a_{v v_0}\cdot a_{u_0 v_1} &= 0, \qquad&\forall\, v_0\xleftrightarrow{}v_1\in\GG.
         \end{align}
     If $u\notin\sib(v)$, considering $u_0= v_0 = u$ in \cref{eq:A:sys:1} yields $a_{v u}\cdot a_{u u}=a_{v u}=0$, where we used the fact that $a_{u u}=1$ as a result of  \cref{lem:A:matrix}. Again, from \cref{lem:A:matrix},
    writing $a_{vu}$ explicitly, we get 
     \begin{equation}
     \label{eq:B:line}
         b_{vu}=\sum_{w\in\\\pa(v)\cap\de(u)}\tilde{\lambda}_{wv}b_{wu}=[(B_\Lambda)_{\pa(v),u}]^T\cdot \Tilde{\lambda}_{\pa(v),v}.
     \end{equation}
    Now, let $u\in\sib(v)$, and let $w\in\sib(u)\setminus\sib(v)$. Considering \cref{eq:A:sys:2} with $u_0=v_1=w$, and $v_0=u$ we get $a_{vu}\cdot a_{ww}=a_{vu}=0$. Proceeding as above, this yields that $u\in\an(v)$ leads to \cref{eq:B:line}, as claimed.

    For the reverse implication,  consider $\tilde{\Lambda}\in\R^{\GG_D}$ such that each one of its column vectors $\tilde{\lambda}_{\pa(v), v}$ is a solution of \cref{eq:lambda:system}. Define $\tilde{\varepsilon}:=A\cdot\varepsilon$, where the matrix $A$ is defined in \cref{eq:A:matrix}. By the definition of $\tilde{\varepsilon}$, we have $\Phi_\GG(\Lambda,\varepsilon)\stackrel{d}{=}\Phi_\GG(\Tilde{\Lambda},\tilde{\varepsilon})$, so it remains to prove  $\tilde{\varepsilon}\in\mathcal{M}(\GG_B)$, that is, $\tilde{\varepsilon}$ satisfies the connected set Markov property with respect to $\GG_{B}$.
    Let $C\subseteq V$. Since $A_{vu} = 0$ whenever $u\notin R_v$, we have
    $$\tilde{\varepsilon}_C=A_{C,D(C)}\cdot\varepsilon_{D(C)},$$
    where $D(C)$ is the set of non-removable ancestors of $C$, so    $$\quad
    \quad D(C):=C\cup\bigcup_{v\in C}(\an(v)\setminus R_v)\subseteq\Sib(C).$$
    Here, the last set inclusion comes from the definition of $R_v${, indeed if $u\in\an(v)$ and $u\notin\Sib(v)$ then $u$ is in $R_v$}. Finally, to prove that $\tilde{\varepsilon}$ satisfies the connected set Markov property, we need to show that $\tilde{\varepsilon}_C\indep \tilde{\varepsilon}_{V\setminus\Sib(C)}$ whenever $C$ is a connected subset of $\GG_B$. For this, it suffices to show that $\varepsilon_{D(C)}\indep \varepsilon_{D(V\setminus\Sib(C))}$.  We will argue that this is indeed the case for $C$ connected by showing i) $D(C)$ is connected and ii) $D(V\setminus\Sib(C))\subseteq V\setminus\Sib(D(C))$. The asserted result will then follow from the fact that $\varepsilon$ satisfies the connected set Markov property with respect to $\GG_{B}$. 
    
    i) To show that $D(C)$ is connected consider $u,v\in C$.  From the definition of $D(C)$, one can see that there are $u_0,v_0\in C$ such that $u_0\in\Sib(u)$ and $v_0\in\Sib(v)$. Because $C$ is connected, a bidirected path joining $u_0$ and $v_0$ exists over $C$, and we can extend the path to suitably join $u$ and $v$ as well.
    
    ii) Notice that $\Sib(D(C))\subseteq\Sib(C)$. This is because if $w\in\Sib(D(C))$, either $w\in\Sib(C)$ or there is $v\in C$ such that $w\in \Sib(\an(v)\setminus R_v)$ which again implies that $w\in\Sib(C)$ by the definition of $R_v$.
    This implies that in order to prove $D(V\setminus\Sib(C))\subseteq V\setminus\Sib(D(C))$, we only need to show  $D(V\setminus\Sib(C))\cap\Sib(C)=\emptyset$.
    Suppose there exists $u\in D(V\setminus\Sib(C))\cap\Sib(C)$, then there are $v\in C$ and $w\in V\setminus\Sib(C)$ such that $v\xleftrightarrow{} u\xleftrightarrow{} w\in\GG_B$, and $u\in R_w$. By the definition of $R_w$, this implies $v\in\sib(w)$ which is impossible as $w\in V\setminus\Sib(C)$. This concludes the proof.
\end{proof}

\begin{definition}
    Let $v\in V$, and let $Q\subseteq (\pa(v)\cup\{v\})$. We define the $v$-rank of $Q$ as
    \begin{equation}
    \label{eq:degree}
        r^v_Q:=\max_{1\leq k\leq |Q|}\{\, (I,P)\in2^{R_v}\times2^{Q}\::\: |I|=|P|=k, \,\, \tilde{\PP}(I,P)\neq\emptyset\},
    \end{equation}
    where $2^{S}$ denotes the power set of $S$.  Recall that $\tilde{\PP}(I,P)$ is the set of non-intersecting systems of directed paths from $I$ to $P$.
\end{definition}
Notice that from \cref{eq:degree} it is immediate that 
\begin{equation}
\label{eq:lowbound}
r^v_{\pa(v)\setminus Q}\geq r^v_{\pa(v)}-|Q|.
\end{equation}
{Theorem~\ref{thm:main:loc} below,} which constitutes our main identifiability result, shows that this lower bound for $r^v_{\pa(v)\setminus Q}$ is reached if and only if $\lambda_{Q,v}$ is generically identifiable.

{To prove our main theorem, we rely on a classical result from combinatorics—the Gessel-Viennot-Lindström lemma—which provides a combinatorial expression for subdeterminants of the matrix $B_{\Lambda}$. For the reader’s convenience, we present a reformulation of this result using our notation; cf.~\citet[Lem.~3.3]{sullivant:2010}, \citet[Supplement][Lem.~1]{foygel:2012}.}

\begin{lemma}[Gessel-Viennot-Lindström lemma]
{Let $\GG_D$ be any directed graph, and let $I, J$ be two subsets of $V$ of the same size. Then for every $\Lambda\in\R^{\GG_D}_{\mathrm{reg}}$ we have:
    $$\det\big((B_{\Lambda})_{I, J}\big) = \det\big((I-\Lambda)^{-1}_{I,J}\big)=\sum_{\Pi\in\mathcal{P}(I,J)}|\sigma_\Pi|\lambda^\Pi=\sum_{\Pi\in\tilde{\mathcal{P}}(I,J)}|\sigma_\Pi|\lambda^\Pi,$$
    where $|\sigma_\Pi|$ denotes the sign of the permutation and $\lambda^\Pi = \prod_{u\to v\in\Pi}\lambda_{uv}$. In particular, $\tilde{\PP}(I,J)=\emptyset$ implies $\det\big((B_{\Lambda})_{I, J}\big) = 0$. The reverse implication holds for a \emph{generic} choice of $\Lambda\in~\R^{\GG_D}_{\mathrm{reg}}$, i.e., for any $\Lambda$ outside a Lebesgue measure 0 subset of $\R^{\GG_D}_{\mathrm{reg}}$}.
\label{lem:gvl}
\end{lemma}

{We now use the Gessel-Viennot-Lindström lemma to characterize} the linear subspace of the solution set of \cref{eq:lambda:system}, which describes $\mathrm{Proj}_{\R^\GG}(\Phi_{-1}(X))$.
\begin{theorem}
\label{thm:main:loc}
Let $v\in V$, and let $Q\subseteq\pa(v)$. The vector $\lambda_{Q,v}$ is generically identifiable if and only if $r^v_{\pa(v)\setminus Q}=r^v_{\pa(v)}-|Q|$, where $r^v_Q$ is defined in \cref{eq:degree}.
\end{theorem}
\begin{proof}
    The vector $\lambda_{Q,v}$ is identifiable if and only $\Tilde{\lambda}_{Q,v}=\lambda_{Q,v}$, for every $\Tilde{\Lambda}\in\mathrm{Proj}_{\R^\GG}(\Phi^{-1}_{\GG}(X))$. We know from \cref{lem:main} that $\tilde{\lambda}_{\pa(v), v}$ is a solution of the linear system given in \cref{eq:lambda:system} for every such matrix $\Tilde{\Lambda}$. Hence, if we define
    \begin{equation*}
    \begin{aligned}
        S^{v}&:=\{\tilde{\lambda}_{\pa(v), v}\in\R^{|\pa(v)|}\::\: [(B_\Lambda)_{\pa(v),R_v}]^T\cdot \Tilde{\lambda}_{\pa(v),v} =[(B_\Lambda)_{v,R_v}]^T\}, \\
        S^v_{Q}&:=\{\tilde{\lambda}_{\pa(v), v}\in S^v\::\: \tilde{\lambda}_{Q,v}=\lambda_{Q,v}\},
    \end{aligned}
    \end{equation*} then $\lambda_{Q,v}$ is identifiable if and only if $S^v_Q=S^v$. By definition, $S^v_{Q}$ is a linear subspace of $S^v$, so the two are equal if and only if they have the same dimension.
    
    We can write $S^v_{Q}$ as the solution space of the following linear system
    \begin{equation}
    \label{eq:lin:sys:1nod}
        \underbrace{\begin{bmatrix}
            (I_p)_{Q,[p]}\\[3pt]
            (B_\Lambda)^v\\
            \end{bmatrix}}_{(B_\Lambda)^v_Q}\cdot \Tilde{\lambda}_{\pa(v),v} =\begin{bmatrix}
            \lambda_{Q,v}\\[3pt]
            [(B_\Lambda)_{v,R_v}]^T\\
        \end{bmatrix},
    \end{equation}
    where $I_p$ is the $p\times p$ identity matrix, and $(B_\Lambda)^v$ is defined in \cref{eq:lambda:system}. We know that the solution space of \cref{eq:lin:sys:1nod} is not empty since $\lambda_{Q,v}$ belongs to it. Hence, we have $\dim(S^v_{Q})=|\pa(v)|-\rank((B_\Lambda)^v_Q)$, which implies  $$\dim(S^v_{Q})=\dim(S^v)\iff\rank((B_\Lambda)^v_Q)=\rank((B_\Lambda)^v).$$ From the definition of $(B_\Lambda)^v_Q$ in \cref{eq:lin:sys:1nod}  one can 
    see that $$\rank((B_\Lambda)^v_Q)=\rank([(B_\Lambda)^v_{R_v, \pa(v)\setminus Q}])+|Q|=\rank([(B_\Lambda)_{\pa(v)\setminus Q,R_v}]^T)+|Q|.$$ 
    Finally, we have 
    \begin{equation*}
        \begin{aligned}
        \dim(S^v_{Q})=\dim(S^v)\iff\rank([(B_\Lambda)_{\pa(v)\setminus\ Q,R_v}]^T)=\rank([(B_\Lambda)_{\pa(v),R_v}]^T)-|Q|,
        \end{aligned}
    \end{equation*}
    which concludes the proof by noticing that from \cref{lem:gvl}, we have $r^v_Q$ is \emph{generically} equal to $\rank([(B_\Lambda)_{Q,R_v}]^T)$ for every $Q\subseteq\pa(v)$.
\end{proof}

\begin{example}
\label{ex:fig2-again}
Consider again the graph in \cref{fig:4:nodes:ex}, as in \cref{ex:fig2}. We have $R_{v_2} = \emptyset$, implying that the parameter $\lambda_{v_1v_2}$ is not  identifiable. In contrast $R_{v_4} = \{v_1, v_2\}$, and there is a system of \emph{non-intersecting} directed paths from $R_{v_4}$ to $\pa(v_4)=\{v_2, v_3\}$ given by $\pi_1 = (v_1,  v_3)$ and $\pi_2 = (v_2, v_2)$. This implies that the vector $\lambda_{\pa(v), v}$ is identifiable.
\end{example}

\begin{example}
\label{ex:par:id}
\begin{figure}[h]
    \begin{tikzpicture}
            \node[draw, circle, inner sep=3pt] (A) at (-4.5,0) {$v_1$};
            \node[draw, circle, inner sep=3pt] (B) at (-2,1) {$v_2$};
            \node[draw, circle, inner sep=3pt] (C) at (-2,-1) {$v_3$};
            \node[draw, circle, inner sep=3pt] (D) at (0.5,0) {$v_4$};
            
            \draw[->] (B) -- (D);
            \draw[->] (C) -- (D);

            \draw[dashed][<->] (A) to (B);
            \draw[dashed][<->] (B) to[bend left=30] (D);
            \draw[dashed][<->] (C) to[bend right=30] (D);            
        \end{tikzpicture}
        
    \caption{{An ADMG for which only one of the causal effects is identifiable.}}
    \label{fig:par:id}
\end{figure}
Consider the graph in \cref{fig:par:id}. We have $R_{v_4} = \{v_2\}$, which has cardinality 1, hence $r^{v_4}_{\pa(v_4)} = 1$. This implies that $\lambda_{\pa(v_4), v_4}$ is not identifiable. However, since there are no directed paths from $v_2$ to $v_3,$ we have $r^{v_4}_{\pa(v_4)\setminus\{v_2\}} = 0 = r^{v_4}_{\pa(v_4)} - |\{v_2\}|$, proving that the parameter $\lambda_{v_2, v_4}$ is identifiable.
\end{example}
As a consequence of an intervention (or of previous knowledge) on the system, it is sometimes possible to assume that some of the parameters of the matrix $\Lambda$ are already known. In the next theorem, we characterize which parameters can be identifiable, assuming partial knowledge of the matrix $\Lambda$.

\begin{theorem}
\label{thm:part:kn}
Let $v\in V$, and let $Q, K\subseteq \pa(v)$. The vector $\lambda_{Q,v}$ is generically identifiable assuming knowledge of $\lambda_{K, v}$ if and only if $r^v_{\pa(v)\setminus (K\cup Q)} = r^v_{\pa(v)\setminus Q} - (|K| - |Q\cup K|).$
\end{theorem}
\begin{proof}
We proceed as in the proof of \cref{thm:main:loc}. Define 
    \begin{equation*}
    \begin{aligned}
        S^{v}_{K}&:=\{\tilde{\lambda}_{\pa(v), v}\in\R^{|\pa(v)|}\::\: [(B_\Lambda)_{\pa(v),R_v}]^T\cdot \Tilde{\lambda}_{\pa(v),v} =[(B_\Lambda)_{v,R_v}]^T, \tilde{\lambda}_{K,v}=\lambda_{K,v}\}, \\
        S^v_{Q, K}&:=\{\tilde{\lambda}_{\pa(v), v}\in S^v_{K}\::\: \tilde{\lambda}_{Q,v}=\lambda_{Q,v}\}.
    \end{aligned}
    \end{equation*}
    Then, $\lambda_{Q, v}$ is identifiable using the knowledge of $\lambda_{K, v}$ if and only if $\dim(S^v_{Q, K}) = \dim(S^v_{K})$. From \cref{eq:lin:sys:1nod}, we know that 
    \begin{equation*}
            \dim(S^v_{Q, K}) = r^v_{\pa(v)\setminus{Q\cup K}} + |Q\cup K|,\quad
            \dim(S^v_{Q}) = r^v_{\pa(v)\setminus{Q}} + |Q|.        
    \end{equation*}
    The result follows by equating the right-hand sides of the equations above.
\end{proof}

The following theorem characterizes the situations in which the matrix $\Lambda$ is identifiable.
\begin{theorem}
\label{thm:main:glob}
    The matrix $\Lambda$ is generically identifiable if and only if for every node $v\in V$, there is a subset $I_v$ of $R_v$ of size $|\pa(v)|$ such that there is a system of \emph{non-intersecting} directed paths from $I_v$ to $\pa(v)$.
\end{theorem}
\begin{proof}
    The matrix $\Lambda$ is identifiable if and only if all of its columns are, so we get the statement by applying \cref{thm:main:loc} to each of the columns, with $Q=\pa(v)$.
\end{proof}

\begin{remark}
    It is noteworthy that \cref{ass:error} is used only for proving the direct implication of \cref{lem:main}. This implies that the necessity of the graphical condition in \cref{thm:main:glob} also holds if the model was extended by not requiring \cref{ass:error} to hold.    
\end{remark}

\begin{remark}
    A direct consequence of \cref{lem:main} is that if the matrix $\Lambda$ is not generically identifiable, the fiber $\mathrm{Proj}_{\R^\GG}(\Phi_{\GG}^{-1}(\Phi_\GG(\Lambda,\varepsilon)))$ has infinite cardinality. This implies that in our setting, there are no ADMGs that are $k$-to-one with finite $k>1$. This is in contrast with the linear Gaussian case; see e.g., \citet[Ex.~8]{foygel:2012}.
\end{remark}

\begin{example}[Non-generic distribution]
\begin{figure}
    \centering  
        \begin{tikzpicture}
            \node[draw, circle, inner sep=4pt] (A) at (-4.5,0) {$v_1$};
            \node[draw, circle, inner sep=4pt] (B) at (-2,0) {$v_2$};           
            \node[draw, circle, inner sep=4pt] (C) at (0.5,0) {$v_3$};
            
            \draw[->] (A) -- (B);
            \draw[->] (B) -- (C);
            \draw[->] (A) to[bend right = 30] (C);
            
            \draw[dashed][<->] (A) to[bend left=30] (B);
            \draw[dashed][<->] (A) to[bend left=30] (C);
        \end{tikzpicture}    
    \caption{Double confounder graph.}
    \label{fig:dc}
\end{figure}
{For the ``double confounder'' graph from  \cref{fig:dc}, all parameters are identifiable. Indeed,  $R_{v_2} = \{v_1\} = \pa(v_2)$ and $R_{v_3} = \{v_1, v_2\} = \pa(v_3)$, so we can apply \cref{thm:main:glob}. Consider the matrix 
\begin{equation*}
(B_{\Lambda})^{v_3} =\setlength{\arraycolsep}{6pt}
\begin{bmatrix}
1 & \lambda_{01}\\[3pt]
\lambda_{01}\lambda_{12} + \lambda_{02} & 1
\end{bmatrix}.
\end{equation*}
We have $\det(B_{\Lambda})^{v_3} = 1 - \lambda_{01}(\lambda_{01}\lambda_{12} + \lambda_{02})$. Hence, the causal effect from $(v_1,v_2)$ on $v_3$ is non-identifiable if and only if the parameters satisfy the following equation $$ \lambda_{01}(\lambda_{01}\lambda_{12} + \lambda_{02}) = 1.$$}
\end{example}
\begin{remark}[Comparison to identification through OICA]
\citet[Thm.~3.9]{tramontano:2024:icml} provide a graphical criterion based on OICA for generic identifiability in lvLiNGAM models. We provide a complete discussion of the connection between the two results in the supplementary material \citep[\S 3]{tramontano:2024:supp}. In a nutshell, by imposing additional assumptions, the lvLiNGAM model allows to identify more edges, but this comes at a price of reduced expressivity of the model as well as issues in parameter estimation.
\end{remark}

\section{Certifying Identifiability}
\label{sec:cert:ident}
Verifying directly whether the condition of \cref{thm:main:loc} is satisfied can be computationally challenging. Following the approach of \cite{brito:2004} and \cite{foygel:2012},  we now introduce an alternative approach that can verify the identifiability condition of \cref{thm:main:loc} in polynomial time in the size of the graph via a maximum flow reformulation. 

For the sake of completeness, we first revisit the definition of the maximum flow problem; further details are available in \citet[\S26]{cormen:2009}. Subsequently, we introduce our reformulation of the identifiability criterion.

The proofs for this section can be found in \citet[\S2.1]{tramontano:2024:supp}.

\subsection{The Maximum Flow Problem}

Let $G=(V, D)$ be a directed graph with source node $s\in V$ and sink node $t\in V$. Let $c_V:V\to\mathbb{R}_{\geq0}$ be a node capacity function, and let $c_D:D\to\mathbb{R}_{\geq0}$ be an edge capacity function. A \emph{flow} on $G$ is a function $f:D\to\mathbb{R}_{\geq0}$ satisfying 
\begin{equation}
\label{eq:flow}
    \begin{aligned}
        \sum_{w\in \ch(v)}f(v,w) = \sum_{u\in \pa(v)}f(u,v)&\leq c_V(v), \quad \forall v\in V\setminus\{s,t\},\\
        f(u,v)&\leq c_D(u,v), \quad \forall u\to v\in D.
    \end{aligned}
\end{equation}
The size of a flow $f$ is defined as 
\begin{equation}
\label{eq:size}
    |f|:=\sum_{w\in \ch(s)}f(s,w)=\sum_{u\in \pa(t)}f(u,t).
\end{equation}
The max-flow problem on $(G,s,t,c_V,c_D)$ is the problem of finding a flow $f$ whose size $|f|$ is maximal.

\subsection{Deciding Generic Identifiability}
\label{subsec:dec:iden}
For every node $v\in V$ and every $Q\subseteq\pa(v)$, let $G^v_Q=(V^v_Q,E^v_Q)$ be defined as follows:
\begin{equation*}
    \begin{aligned}
    V^v_Q:=&\an(v){\setminus\{v\}}\cup\{s_v,t_v\},\\
    E^v_Q:=&\{s_v\to u\::\: u\in R_v\}\cup\{u\to t_v\::\: u\in Q\}\cup\{u\to w\::\: u\to w\in\GG\},
    \end{aligned}
\end{equation*}
where $s_v$ and $t_v$ are, respectively, newly introduced source and sink nodes.
The edge capacity is $\infty$ for all the edges. The node capacity is $\infty$ for both the sink and the source, and $1$, otherwise.   We denote the maximum size of any flow on $G^v_Q$ by $\,\maxflow{(G^v_Q)}$.

\begin{lemma}
\label{lem:mflow}
It holds that $\maxflow{(G^v_Q)}=r^v_Q$.
\end{lemma}
{Using the above lemma, we can now reformulate \cref{thm:main:loc} in a form that is provably certifiable in polynomial time.}
\begin{theorem}
\label{thm:cert:loc}
Given a mixed graph $\GG=(V,\Ed,\Eb)$, a node $v\in V$, and any $Q\subseteq\pa(v)$, the generic identifiability of $\lambda_{Q,v}$ holds if and only if $\maxflow{(G^v_Q)}=|Q|$, which can be certified in $\mathcal{O}(|V|^{2+o(1)})$ time.
\end{theorem}

\begin{theorem}
\label{thm:cert:glob}
Given a mixed graph $\GG=(V,\Ed,\Eb)$, the generic identifiability of $\Lambda$ holds if and only if $\maxflow{(G^v_{\pa(v)})}=|\pa(v)|$ for all $v\in V$, which can be certified in $\mathcal{O}(|V|^{3+o(1)})$ time. 
\end{theorem}
\begin{remark}
The complexity statement in the above theorem relies on the recent algorithm proposed by \citet{chen:2022}. However, due to the large constants involved in its asymptotic analysis, the advantages of this algorithm would only appear for graph sizes far exceeding those considered in our experiments (and in typical causal inference settings). Therefore, in our implementation, we employ the standard algorithm of \citet{dinitz:1970}, which, despite its higher asymptotic complexity, achieves satisfactory performance in practice.
\end{remark}
\begin{example}
\cref{fig:max:flow} illustrates the maximum flows when applying the criterion from \cref{thm:cert:loc} to the nodes $v_2$ and $v_4$ of the ADMG in \cref{fig:4:nodes:ex}.
\medskip

$G^{v_2}_{\pa(v_2)}:$  The graph is constructed for parameter $\lambda_{12}$.  The only flow on  $G^{v_2}_{\pa(v_2)}$ is the trivial flow setting all edges to $0$. Hence, $\lambda_{12}$ is not identifiable. 
\medskip

$G^{v_4}_{\pa(v_4)}:$  The graph is constructed for parameter $\Lambda_{\{2, 3\}, 4}$.  The figure displays a flow on $G^{v_4}_{\pa(v_4)}$ of  
 size $|\pa(3)|=2$.  Consequently, the parameters $\lambda_{24}$ and $\lambda_{34}$ are identifiable.
\end{example}

    \begin{figure}[t]
        \centering
        \begin{tikzpicture}
            \node at (-0,-0.5) {$G^{v_2}_{\pa(v_2)}:$};
            \node[draw, inner sep=4pt] (S1) at (1.5,1) {$s_{v_2}$};
            \node[draw, inner sep=4pt] (P01) at (1.5,-0.5) {$v_1$};
            \node[draw, inner sep=4pt] (T1) at (1.5,-2) {$t_{v_2}$};
            
            \draw[->] (P01) --node[left] {0} (T1);
            
            \node at (5,-0.5) {$G^{v_4}_{\pa(v_4)}:$};
            \node[draw, inner sep=4pt] (S3) at (8,1) {$s_{v_4}$};
            \node[draw, inner sep=4pt] (P03) at (6.5,-0.5) {$v_1$};
            \node[draw, inner sep=4pt] (P13) at (9.5,-0.5) {$v_2$};
            \node[draw, inner sep=4pt] (P23) at (7.3,-2) {$v_3$};
            \node[draw, inner sep=4pt] (T3) at (8,-3.5) {$t_{v_4}$};

            \draw[->] (S3) --node[left] {1\,\,} (P03);
            \draw[->] (S3) --node[right] {\,\,1} (P13);
            \draw[->] (P03) --node[below] {0} (P13);
            \draw[->] (P03) --node[left] {1\,\,} (P23);
            \draw[->] (P13) --node[right] {1\,} (T3);
            \draw[->] (P23) --node[left] {1} (T3);
            
        \end{tikzpicture}
        
        \caption{Two maximum flow problems corresponding to the ADMG of \cref{fig:4:nodes:ex}.}
        \label{fig:max:flow}
    \end{figure}

\section{The Genericity Condition for the Error Distribution}
\label{sec:gen:ass}
The idea underlying \cref{ass:error} is that it should not be possible to linearly disentangle a general dependence between two errors $\varepsilon_u$ and $\varepsilon_v$.  In other words, if two different linear combinations of $\varepsilon$ are independent, then at least one of them cannot have any signal coming from $(\varepsilon_u,\varepsilon_v)$.  
The purpose of this section is to prove that this fact is indeed true for two tractable subfamilies of joint distributions for the errors. 
Specifically, \cref{subsubsec:lin:lat} considers the setting in which dependence is generated through linear latent  factor models,  and \cref{subsubsec:fin:mom} treats  distributions with finite moments.

\subsection{Linear Factor Models}
\label{subsubsec:lin:lat}

Assume that the error vector $\varepsilon$ is generated according to a sparse factor model that respects the Markov property of the bidirected part $\GG_B$ of a given ADMG $\GG$.  Define a latent factor graph for $\GG_B$ to be any DAG $\mathcal{L}=(V\cup L, E_\mathcal{L})$, in which the latent nodes $L$ are source nodes and whose latent projection \citep[see][Sec.~3]{verma:1999} on the nodes in $V$ is equal to $\GG_B$.
Define $\MM(k)$ to be the set of $k$-dimensional random vectors with independent and non-Gaussian components.  Then, the sparse factor model associated to $\mathcal{L}$ is the set of random vectors
\begin{equation}
\label{eq:lat:proj:model}
\mathcal{M}^{\mathcal{L}}(\GG_B)=\{\varepsilon\in\mathcal{M}(\GG_B)\::\:\exists\,  \eta\in\MM(|V|+|L|),\,\, H\in\R^{\mathcal{L}},\,\, \varepsilon=H_{L,V}^T\cdot \eta_L+\eta_V\}.
\end{equation}

\begin{theorem}
\label{thm:latent:proj}
Let $\mathcal{L}=(V\cup L, E_\mathcal{L})$ be a latent factor graph for $\GG_B$, and for any subset $C\subset V$ define  $L_C:=\{l\in L\::\: \ch_{\mathcal{L}}(l)\subseteq C\}$.
If for every edge $u\xleftrightarrow{} v\in\GG_B$ there is a clique ${C_{u,v}\supseteq\{u, v\}}$ (a subset of $V$ for which every pair of nodes is adjacent) in $\GG_B$ such that $|L_{{C_{u,v}}}|\geq |{C_{u,v}}|-1$ then $\varepsilon$ satisfies \cref{ass:error} for Lebesgue-almost every matrix $H\in\R^\mathcal{L}$.
\end{theorem}

\begin{proof}
Let $a_1,a_2\in\mathbb{R}^V$, and consider $\varepsilon=H_{L,V}^T\cdot \eta_L+\eta_V$ as in \cref{eq:lat:proj:model}.  Applying the Darmois-Skitovich theorem \citep[Thm.~9.5]{comon:jutten:handbook} to a pair of independent linear transformation of $\varepsilon$, namely $a_1^T\varepsilon$ and $a_2^T\varepsilon$, we obtain
\begin{align}
    \label{eq:same:col:cond}
    a_{1s}\cdot a_{2s}&=0,\;  &\forall s\in V,\\
    \label{eq:diff:col:cond}
    (a_1^T H_{L,V}^T)_{l}\cdot (a_2^T H_{L,V}^T)_{l}&=0,\; &\forall l\in L.
\end{align}
Note that \cref{eq:same:col:cond} already gives the part of the claim referring to the case $u=v$ in \cref{ass:error}.
It remains to consider the case of two nodes $u,v$ that are adjacent in $\GG_B$.

Let $u\xleftrightarrow{} v\in\GG_B$, and assume for contradiction that $a_{1u}\cdot a_{2v}\neq 0$.
Consider a clique ${C_{u,v}}$ as in the statement of the theorem.  
The vector $a_{C_{u,v}} := (a_{1, C_{u,v}}, a_{2, C_{u,v}})$ is a solution of the following system of quadratic equations:
$$
\left(\sum_{c\in C_{u,v}}a_{1c}H_{cl}\right)\cdot\left(\sum_{c\in C_{u,v}}a_{2c}H_{cl}\right) = 0,\,\quad l\in L_{C_{u,v}};
$$
we denote the system by $\mathcal{S}_{C_{u,v}}$.
Notice that from \cref{eq:same:col:cond} we know that the vector $a_{C_{u,v}}$ has at most $|C_{u,v}|$ non-zero entries. We now show that, for a generic choice of the entries of $H$, $\mathcal{S}_{C_{u,v}}$ does not admit solutions with $a_{1u}\cdot a_{2v}\neq 0$. Following the case distinctions resulting from the vanishing of the first or the second factor in the equations in~\eqref{eq:same:col:cond}, the solution set of $\mathcal{S}_{C_{u,v}}$ can be written as the union of the solution set of $2^{|L_{C_{u,v}}|}$ homogeneous linear systems. Each of these linear systems can be characterized by a partition of $L_{C_{u,v}}$ defined as follows:
$$
L_1 :=\{l\in L_{C_{u,v}}\::\: (a_1^T H_{L,V}^T)_{l} = 0\},\,\quad L_2 := L_{C_{u,v}}\setminus L_1.
$$
We denote by $\mathcal{S}_1$ and $\mathcal{S}_2$ the linear systems associated to $L_1$ and $L_2$, respectively.
Define $V_1 = \{v\in C_{u,v}\::\: a_{1v} = 0\}$ and $V_2 = \{v\in C_{u,v}\::\: a_{2v} = 0\}$.  

If $V_1\cap V_2\neq\emptyset$, the vector $a_{C_{u,v}}$ has at most $|C_{u,v}|-1$ non-zero entries, implying that $\mathcal{S}_1\cup\mathcal{S}_2$ has $|L_{C_{u,v}}|$ equations and $|C_{u,v}|-1$ parameters. If $|L_{C_{u,v}}|\geq|C_{u,v}|-1$, for a generic choice of the entries of $H$, such a system admits only the 0 solution \citep[Lemma]{okamoto:1973}. Hence, the assumption that $a_{1u}\cdot a_{2v}\neq 0$ leads to a contradiction.

For $V_1\cap V_2 = \emptyset$, we now show that either $\mathcal{S}_1$ or $\mathcal{S}_2$ admits only the 0 solution. Notice that since $a_{1u}\cdot a_{2v}\neq 0$ we have $V_1, V_2\neq\emptyset$. This implies that both $\mathcal{S}_1$ and $\mathcal{S}_2$ can have a non-zero solution for a generic choice of the entries of $H$ only if 
$$
|L_1|\leq |V_1|-1, \,\quad |L_2|\leq |V_2|-1.
$$
This would lead to $|L_{C_{u,v}}| = |L_1| + |L_2|\leq|V_1|+|V_2|-2 = |C_{u,v}|-2$, which contradicts the hypothesis that $|L_{C_{u,v}}|\geq|C_{u,v}|-1$.
\end{proof}

\begin{corollary}
Let $\mathcal{L}(\GG_B)$, be the canonical DAG associated to $\GG_B$ \citep[\S6]{richardson:2002} then \cref{ass:error} is satisfied for a generic choice of parameters of $\mathcal{M}^{\mathcal{L}(\GG_B)}(\GG_B)$.
\end{corollary}

\begin{example}
\label{ex:faith:viol}
We now present examples that on one hand illustrate the previous results and on the other hand emphasize the subtleties of involved genericity assumptions.  To this end,
we considered the graph in \cref{fig:faith:fail} that appears in \citet[Fig.~5]{foygel:2022}.

Notice that in the proof of our main result, \cref{ass:error} is used only for matrices with a specific structure, described in \cref{lem:A:matrix}. Therefore,  we focus on this type of matrices. In particular, we will consider the matrix 
\begin{equation}
\label{eq:A:ex}
    A=\begin{pmatrix}
    a_1^T \\ a_2^T
    \end{pmatrix}=\begin{pmatrix}
        1 & 0 & 0 & 0 & 0 \\
        0 & 0 &  a_{53} & a_{54} & 1
    \end{pmatrix},
\end{equation}
and the bidirected graph $\GG_B$, corresponding to $\GG$ in \cref{fig:faith:fail} with respect to the latent factor models $\mathcal{L}, \mathcal{L}_1, \mathcal{L}_2$ given in \cref{fig:faith:fail}, and \cref{fig:faith:sat}.
\begin{enumerate}
    \item[(i)] First, we highlight a case in which our theorem does not apply, and indeed \cref{ass:error} does not hold generically. For this, consider the pair, $v_2\xleftrightarrow{}v_3\in\GG_B$, the only latent parent of both in $\mathcal{L}$ is $l_1$ and $\ch(l_1)=\{v_1, v_2, v_3, v_4\}$. This means that the only clique we can consider is $C_{v_2v_3}=\{v_1, v_2, v_3, v_4\}$ and $|L_{C_{v_2v_3}}|=1$, hence the condition in \cref{thm:latent:proj} is violated. Now, we will show that \cref{eq:diff:col:cond} has a nonzero solution. Indeed, the only latent variable for which the system is not trivially satisfied is $l_1$, implying that any solution of the equation $a_{53}H_{3l_1}+a_{54}H_{4l_1}=0$, is also a solution of \cref{eq:diff:col:cond}. 
    \begin{figure}[t]
        \centering
        \begin{tikzpicture}[scale=0.8]
            \node at (-6,0) {$\mathcal{{L}}:$};
            \node[draw, circle, inner sep=4pt] (A) at (-5,0) {$v_1$};
            \node[draw, circle, inner sep=4pt] (B) at (-3.5,0) {$v_2$};
            \node[draw, circle, inner sep=4pt] (C) at (-2,0) {$v_3$};
            \node[draw, circle, inner sep=4pt] (D) at (-0.5,0) {$v_4$};
            \node[draw, circle, inner sep=4pt] (E) at (1,0) {$v_5$};
            
            \node[draw, circle, fill=gray!40, inner sep=4pt] (F) at (-2.75,1.5) {$l_1$};
            \node[draw, circle, fill=gray!40, inner sep=4pt] (G) at (0.25,1.5) {$l_2$};
            \node[draw, circle, fill=gray!40, inner sep=4pt] (H) at (-0.5,-1.5) {$l_3$};
            
            \draw[->] (D) -- (E);
            \draw[->] (C) to[bend right = 30] (E);
            
            \draw[dashed][->] (F) -- (A);
            \draw[dashed][->] (F) -- (B);
            \draw[dashed][->] (F) -- (C);
            \draw[dashed][->] (F) -- (D);
            
            \draw[dashed][->] (G) -- (D);
            \draw[dashed][->] (G) -- (E);
            
            \draw[dashed][->] (H) -- (C);
            \draw[dashed][->] (H) -- (E);
                  
            \node at (2.5,0) {${\mathcal{G}}:$};
            \node[draw, circle, inner sep=4pt] (A1) at (3.5,0) {$v_1$};
            \node[draw, circle, inner sep=4pt] (B1) at (5,0) {$v_2$};
            \node[draw, circle, inner sep=4pt] (C1) at (6.5,0) {$v_3$};
            \node[draw, circle, inner sep=4pt] (D1) at (8,0) {$v_4$};
            \node[draw, circle, inner sep=4pt] (E1) at (9.5,0) {$v_5$};
          
            \draw[->] (D1) -- (E1);
            \draw[->] (C1) to[bend right = 30] (E1);
            
            \draw[dashed][<->] (A1) to (B1);
            \draw[dashed][<->] (B1) to (C1);
            \draw[dashed][<->] (C1) to (D1);
            \draw[dashed][<->] (A1) to[bend left=30] (C1);
            \draw[dashed][<->] (A1) to[bend left=50] (D1);
            \draw[dashed][<->] (B1) to[bend left=30] (D1);
            \draw[dashed][<->] (C1) to[bend right=50] (E1);
            \draw[dashed][<->] (D1) to[bend left=30] (E1);
        \end{tikzpicture}
    \caption{A latent factor model under which \cref{ass:error} does not hold and the corresponding latent projection.}
    \label{fig:faith:fail}
    \end{figure}
    \item[(ii)]
    Second, we like to stress that our sufficient condition is not necessary. For an example, we may add a latent node $l_4$ to the graph as in the graph $\mathcal{L}_1$ in \cref{fig:faith:sat}, the condition of \cref{thm:latent:proj} is still not satisfied. However, in this case,  \cref{ass:error} cannot be violated by the matrices described in \cref{eq:A:ex}. To see this, consider \cref{eq:diff:col:cond} for the latent variables $l_1$ and $l_4$, which leads to the following system of equations for $(a_{53}, a_{54})$.
    $$
    \begin{cases}
    H_{1l_1}(a_{53}H_{3l_1}+a_{54}H_{4l_1}) &= 0\\
    H_{1l_4}(a_{53}H_{3l_1}) &= 0,
    \end{cases}
    $$ 
    Clearly, for a generic choice of the matrix $H\in\R^{\LL_2}$, the only solution to this system of equations is $(a_{53}, a_{54})=(0,0)$.
    \begin{figure}[t]
        \centering
        \begin{tikzpicture}[scale = 0.8]
            \node at (-6,0) {$\mathcal{L}_1:$};
            \node[draw, circle, inner sep=4pt] (A) at (-5,0) {$v_1$};
            \node[draw, circle, inner sep=4pt] (B) at (-3.5,0) {$v_2$};
            \node[draw, circle, inner sep=4pt] (C) at (-2,0) {$v_3$};
            \node[draw, circle, inner sep=4pt] (D) at (-0.5,0) {$v_4$};
            \node[draw, circle, inner sep=4pt] (E) at (1,0) {$v_5$};
            
            \node[draw, circle, fill=gray!40, inner sep=4pt] (F) at (-2.75,1.5) {$l_1$};
            \node[draw, circle, fill=gray!40, inner sep=4pt] (G) at (0.25,1.5) {$l_2$};
            \node[draw, circle, fill=gray!40, inner sep=4pt] (H) at (-0.5,-1.5) {$l_3$};
            \node[draw, circle, fill=gray!40, inner sep=4pt] (I) at (-2.75,-1.5) {$l_4$};
            
            \draw[->] (D) -- (E);
            \draw[->] (C) to[bend right = 30] (E);
            
            \draw[dashed][->] (F) -- (A);
            \draw[dashed][->] (F) -- (B);
            \draw[dashed][->] (F) -- (C);
            \draw[dashed][->] (F) -- (D);
            
            \draw[dashed][->] (I) -- (A);
            \draw[dashed][->] (I) -- (C);
            \draw[dashed][->] (I) -- (D);
            
            \draw[dashed][->] (G) -- (D);
            \draw[dashed][->] (G) -- (E);
            
            \draw[dashed][->] (H) -- (C);
            \draw[dashed][->] (H) -- (E);

            \node at (2.5,0) {$\mathcal{L}_2:$};
            \node[draw, circle, inner sep=4pt] (A1) at (3.5,0) {$v_1$};
            \node[draw, circle, inner sep=4pt] (B1) at (5,0) {$v_2$};
            \node[draw, circle, inner sep=4pt] (C1) at (6.5,0) {$v_3$};
            \node[draw, circle, inner sep=4pt] (D1) at (8,0) {$v_4$};
            \node[draw, circle, inner sep=4pt] (E1) at (9.5,0) {$v_5$};
          
            \node[draw, circle, fill=gray!40, inner sep=4pt] (F1) at (5.75,1.5) {$l_1$};
            \node[draw, circle, fill=gray!40, inner sep=4pt] (G1) at (8.75,1.5) {$l_2$};
            \node[draw, circle, fill=gray!40, inner sep=4pt] (H1) at (8.75,-1.5) {$l_3$};
            \node[draw, circle, fill=gray!40, inner sep=4pt] (I1) at (6.5,-1.5) {$l_4$};
            \node[draw, circle, fill=gray!40, inner sep=4pt] (I2) at (4.25,-1.5) {$l_5$};
            
            \draw[->] (D1) -- (E1);
            \draw[->] (C1) to[bend right = 30] (E1);
            
            \draw[dashed][->] (F1) -- (A1);
            \draw[dashed][->] (F1) -- (B1);
            \draw[dashed][->] (F1) -- (C1);
            \draw[dashed][->] (F1) -- (D1);
            
            \draw[dashed][->] (I1) -- (A1);
            \draw[dashed][->] (I1) -- (B1);
            \draw[dashed][->] (I1) -- (C1);
            \draw[dashed][->] (I1) -- (D1);
            
            \draw[dashed][->] (I2) -- (A1);
            \draw[dashed][->] (I2) -- (B1);
            \draw[dashed][->] (I2) -- (C1);
            \draw[dashed][->] (I2) -- (D1);
            
            \draw[dashed][->] (G1) -- (D1);
            \draw[dashed][->] (G1) -- (E1);
            
            \draw[dashed][->] (H1) -- (C1);
            \draw[dashed][->] (H1) -- (E1);
        \end{tikzpicture}
    \caption{Two latent factor models with the same latent projection as in \cref{fig:faith:fail}, under which \cref{ass:error} holds generically, for matrices as in \cref{eq:A:ex}. The graph on the left does not satisfy the condition of \cref{thm:latent:proj} while the right graph satisfies the condition. Hence, the condition introduced in \cref{thm:latent:proj} is sufficient but not necessary.}
    \label{fig:faith:sat}
    \end{figure}
    \item[(iii)] Finally, the graph $\mathcal{L}_2$ provides an example in which the genericity of \cref{ass:error} can be established via \cref{thm:latent:proj}.
\end{enumerate}
\end{example}

\subsection{Random Variables with Finite Moments}
\label{subsubsec:fin:mom}
We now turn to a setting where the error vector has finite moments up to a suitable order.  As we show in \cref{thm:gen:cond} below, the distributions at which \cref{ass:error} fails define a set of moments, or also cumulants, that form a Lebesgue null set in all possible moments/cumulants up to the considered truncation order.  The proofs for the results presented in this section can be found in \citet[\S2.2]{tramontano:2024:supp}.

\begin{definition}
The $k$-th cumulant tensor of a random vector $\varepsilon=(\varepsilon_1,\dots,\varepsilon_p)$ is the $k$-way tensor in $\mathbb{R}^{p\times\dots\times p}\equiv(\mathbb{R}^p)^k$ whose entry in position $(i_1,\dots,i_k)$ is the joint cumulant
\begin{equation*}
    \begin{aligned}
           \mathcal{C}^{(k)}(\varepsilon)_{i_1,\dots,i_k}:=\sum_{(A_1,\dots,A_L)}(-1)^{L-1}(L-1)!\mathbb{E}\bigg[\prod_{j\in A_1} \varepsilon_j\bigg]\cdots\mathbb{E}\bigg[\prod_{j\in A_L} \varepsilon_j\bigg],
    \end{aligned}
\end{equation*}
where the sum is taken over all partitions $(A_1,\dots, A_L)$ of the multiset $\{i_1,\dots,i_k\}$.
\end{definition}

Cumulant tensors are symmetric, i.e., 
\[
\mathcal{C}^{(k)}(\varepsilon)_{i_1,\dots,i_k}
=\mathcal{C}^{(k)}(\varepsilon)_{\sigma(i_1),\dots,\sigma(i_k)} \ \forall\sigma\in S_k, 
\]
where $S_k$ is the symmetric group on $[k]=\{1,\dots,k\}$.  We write $\Sym_k(p)$ for the subspace of symmetric tensors in $(\mathbb{R}^p)^k$.



To justify \cref{ass:error}, we wish to offer statements of its generic validity.  Our strategy to do so in the present context is to consider cumulants up to a suitable truncation order $k$.  In the remainder of this section, we consider a mixed graph $\GG$ with $p$ nodes, which we label by taking the vertex set to be $V=[p]$.

\begin{definition}
\label{def:finite:mom:models}
    Let $\mathcal{M}_\infty(\GG_B)$ be the subset of $\mathcal{M}(\GG_B)$ yielding distributions with all moments finite. For any integer $k\ge 2$, let  
    \[
    \mathcal{M}^{(k)}(\GG_B) = \left\{\mathcal{C}^{(k)}\in\Sym_k(\R^{p})\::\: 
        \mathcal{C}^{(k)}_{i_1,\dots,i_k} = 0 \emph{ if } \{i_1,\dots,i_k\} \emph{ is not connected in } \GG_B\;
        \right\}.
    \]
    Moreover, we let
    \begin{equation*}
            \mathcal{M}^{\leq k}(\GG_B)= \mathcal{M}^{(2)}(\GG_B) \times \cdots \times \mathcal{M}^{(k)}(\GG_B).
    \end{equation*}
\end{definition}

\begin{lemma}
\label{lem:zeros:cum}
Fix any integer $k\ge 1$.  
\begin{itemize}
\item[(i)] The map $\phi^{k}: \mathcal{M}_\infty(\GG_B)\to \mathcal{M}^{(k)}(\GG_B)$ that sends random vectors with all moments finite to their $k$-th cumulant tensors is well-defined, i.e., $\phi^{k}(\mathcal{M}_\infty(\GG_B))\subseteq \mathcal{M}^{(k)}(\GG_B)$.
\item[(ii)] Define the map $\phi^{\leq k}=(\phi^l)_{l\le k}:\mathcal{M}_\infty(\GG_B)\to \mathcal{M}^{\leq k}(\GG_B)$.  Then $\phi^{\leq k}(\mathcal{M}_\infty(\GG_B))$ is a full dimensional subset of $\mathcal{M}^{\leq k}(\GG_B)$.
\end{itemize}
\end{lemma}

\begin{theorem}
\label{thm:gen:cond}
Let $$\kappa(\GG_B):=\{A=(a_{ij})\in\R^{2\times p}\::\: a_{1i}\cdot a_{2j}=0, \emph{ if } u_i\xleftrightarrow{}u_j\in\GG_B \emph{ or } i=j\}.$$
For every $\varepsilon\in\mathcal{M}_{\infty}(\GG_B)$,  define $\kappa(\varepsilon)=\{A\in\R^{2\times p}\::\: (A\varepsilon)_1\indep(A\varepsilon)_2\}$, and let $\mathcal{S}(\GG_B)=\{\varepsilon\in\mathcal{M}_{\infty}(\GG_B)\::\:\kappa(\varepsilon)\setminus\kappa(\GG_B)\neq0\}$, which is precisely the set of distributions for which \cref{ass:error} fails. Then there is a positive integer $k\leq 2(p+1)$ such that $\phi^{\leq k}(\mathcal{S}(\GG_B))$ is a Lebesgue measure 0 subset of $\mathcal{M}^{\leq k}(\GG_B)$.
\end{theorem}

We remark that, for simplicity, we stated \cref{thm:gen:cond} for distributions with finite moments of any order. However, we only needed the first $2(p+1)$  moments to be finite.

\begin{example}
\label{ex:faith:viol:2}
    One simple type of exceptional distribution for which \cref{ass:error} fails to hold is provided by distributions that are obtained as linear transformations of independent non-Gaussian variables. For example, let $U_1, U_2$ be two independent, standard univariate normal distributions, let $V_i = \sqrt[3]{U_i}$ for $i = 1,2$, and  let $X = B\cdot (V_1, V_2)$ for any invertible 2 by 2 matrix $B$.  Then $X\in\MM_{\infty}(\GG_B)$ for $\GG_B=\{\{1,2\},  \{1\xleftrightarrow{} 2\}\}$, but by construction  $V=B^{-1}\cdot X$, and the fact that $V_1\indep V_2$ implies $X\in \mathcal{S}(\GG_B)$.  As noted in \cite{schkoda:2025}, linear transformations of independent components form a null set already when considering cumulants of order up to $k\le 3$.
\end{example}

\begin{remark}
\cref{thm:gen:cond} is of independent interest given the recent scholarly attention to generalizations of ICA that can deal with dependent error terms; see, e.g., \cite{mesters:2022, garrotelopez:2024, wang:2024}. Indeed, if we consider $\GG_B$ to be the empty graph, then \cref{thm:gen:cond} reduces to a generic version of the classical Darmois-Skitovich theorem that underlies ICA theory \citep[Thm.~9.5]{comon:jutten:handbook}.  From this perspective, \cref{thm:gen:cond} provides a generic generalization of the Darmois-Skitovich theorem to the case where the independence structure of the sources is more complex. A consequence of Examples \ref{ex:faith:viol}-\ref{ex:faith:viol:2} is that a global generalization of the Darmois-Skitovich theorem, i.e., one that holds for every non-Gaussian distribution, cannot be achieved.
\end{remark}

\section{Cyclic Graphs}
\label{sec:cycles}
Up to this point, we have exclusively studied acyclic models.
This assumption has allowed us to obtain a complete characterization of the identifiable parameters. 
In this section, we relax this assumption and show that the graphical criterion proposed in \cref{sec:id:result} remains a necessary condition but is no longer sufficient. 
Moreover, we will provide a complete characterization of the identifiable parameters for a special subclass of cyclic graphs.

The first issue that one encounters when dealing with cyclic models is that the matrix $(I-\Lambda)$ might not be invertible. This implies that the assignment $X = \Lambda^T\cdot X + \varepsilon$ does not induce a unique solution for $X$. Hence, we need to restrict our attention to a subset of $\R^{\GG_D}$, namely, the set $\R^{\GG_D}_{\mathrm{reg}}:=\{\Lambda\in\R^{\GG_D}\::\:\det(I-\Lambda)\neq0\}$. In this section, we focus on the identifiability of the matrix $\Lambda$.
\begin{definition}
\label{prob:def:cycles}
    Define the parametrization map
    \begin{equation*}
    \label{eq:phi:map:cycle}
        \begin{aligned}
            \Phi_{{\GG}_{\mathrm{reg}}}: \R^{\GG_D}_{\mathrm{reg}}\times\mathcal{M}(\GG_B)&\xrightarrow{}\MM(\GG)\\
             (\Lambda,\varepsilon) &\mapsto (I-\Lambda)^{-T}\varepsilon,
        \end{aligned}
    \end{equation*}
    and for every $X\in\MM(\GG)$, let the fiber of $X$ with respect to $\Phi_\GG$ be 
    \begin{equation}
    \label{eq:fibers:cycle}
        \Phi_{{\GG}_{\mathrm{reg}}}^{-1}(X):=\{(\Lambda, \varepsilon) \in \R^{\GG_D}_{\mathrm{reg}}\times\mathcal{M}(\GG_B) \::\: \Phi_{{\GG}_{\mathrm{reg}}}(\Lambda, \varepsilon) \stackrel{d}{=} X\}.
    \end{equation}
    For any \emph{generic} choice of $(\Lambda, \varepsilon)$, let  $X=\Phi_\GG(\Lambda, \varepsilon)$.
    We say that the graph $\GG$ is generically identifiable if $\mathrm{Proj}_{\R^\GG_{\mathrm{reg}}}(\Phi_{{\GG}_{\mathrm{reg}}}^{-1}(X)) = \{\Lambda\}$.
\end{definition}

\begin{lemma}
\label{lem:main:cycles}
Let $X=\Phi_{{\GG}_{\mathrm{reg}}}(\Lambda,\varepsilon)$ for a generic choice of parameters $(\Lambda,\varepsilon)\in\R^{\GG_D}_{\mathrm{reg}}\times\mathcal{M}(\GG_B)$. The matrix $\Tilde{\Lambda}\in\R^{\GG_D}_{\mathrm{reg}}$ belongs to $\mathrm{Proj}_{\R^\GG_{\mathrm{reg}}}(\Phi_{{\GG}_{\mathrm{reg}}}^{-1}(X))$ if it is a solution to the following linear system of equations:
\begin{equation}
\label{eq:lambda:system:cycle}
    \underbrace{[(B_\Lambda)_{\pa(v),R_v}]^T}_{(B_\Lambda)^v}\cdot \Tilde{\lambda}_{\pa(v),v} =[(B_\Lambda)_{v,R_v}]^T,\quad \forall v\in V.
\end{equation}
\end{lemma}
\begin{proof}
    It suffices to notice that for the reverse implication of the proof of \cref{lem:main}, we never used the acyclicity of the graph. Hence, the same proof applies.
\end{proof}

\begin{theorem}
\label{thm:main:glob:cycle}    
    If a mixed graph $\GG$ is identifiable, then for every $v\in V$, there exists a subset $I_v$ of $R_v$ with the size $|\text{pa}(v)|$, such that there is a system of non-intersecting directed paths from $I_v$ to $\text{pa}(v)$.
\end{theorem}

\begin{proof}
    If there is a $v\in V$ that satisfies the assumptions of the theorem, then the matrix $(B_\Lambda)^v$ is rank deficient; hence there is $\Tilde{\lambda}_{\pa(v),v}\neq\lambda_{\pa(v),v}$ that solves \cref{eq:lambda:system:cycle}. The matrix $\tilde{\Lambda}$ obtained from $\Lambda$ by substituting the column corresponding to $v$ with $\Tilde{\lambda}_{\pa(v),v}$, belongs to $\mathrm{Proj}_{{\R^\GG}_{\mathrm{reg}}}(\Phi_{\GG}^{-1}(X))$ according to \cref{lem:main:cycles}. Hence, the matrix $\Lambda$ is not identifiable.
\end{proof}
\begin{example}[A non-identifiable cyclic graph]
    For the graph in \cref{fig:non:id:cycle}, we have $\pa(v_2) = \{v_1, v_3\}$, and $I_{v_2} = \{v_1\}$. Considering $v = v_2$ in \cref{thm:main:glob:cycle}, we can see that the matrix $\Lambda$ is not identifiable for this cyclic graph.
    \begin{figure}[t]
        \centering
        \begin{tikzpicture}            
            \node[draw, circle, inner sep=4pt] (A) at (0,0) {$v_1$};
            \node[draw, circle, inner sep=4pt] (B) at (2,0) {$v_2$};
            \node[draw, circle, inner sep=4pt] (C) at (4,0) {$v_3$};
            
            \draw[->] (A) to[bend left = 0] (B);
            \draw[->] (B) to[bend left = 25] (C);
            \draw[->] (C) to[bend left = 25] (B);
            \draw[dashed][<->] (B) to[bend left=60] (C);
            
        \end{tikzpicture}
        \caption{A non-identifiable cyclic graph.}
        \label{fig:non:id:cycle}
    \end{figure}
\end{example}

\begin{example}[Non-sufficiency of the graphical criterion]
\label{ex:non:suff}
    Let $\GG_2$ be the 2-cycle in \cref{fig:two:cycle}. The matrix $A$ of \cref{eq:A:matrix} will have the following form:
    \begin{equation*}
        A = \setlength{\arraycolsep}{6pt} 
        \begin{bmatrix}
    b_{v_1v_1}-\tilde{\lambda}_{v_2v_1}b_{v_2v_1} & b_{v_1v_2}-\tilde{\lambda}_{v_2v_1}b_{v_2v_2} \\[4pt] 
    b_{v_2v_1}-\tilde{\lambda}_{v_1v_2}b_{v_1v_1} & b_{v_2v_2}-\tilde{\lambda}_{v_1v_2}b_{v_1v_2}
    \end{bmatrix}
        .
    \end{equation*}
    From \cref{ass:error} and the fact that there are no bidirected edges in the graph, we have
    \begin{equation*}
        a_{v_1v_1}\cdot a_{v_2v_1} = a_{v_1v_2}\cdot a_{v_2v_2} = 0.
    \end{equation*}
    Since the graph has cycles, we cannot rule out the possibility that the diagonal entries of $A$ are equal to zero. Hence, a valid solution is
    \begin{equation*}
        \tilde{\lambda}_{v_1v_2} = \frac{b_{v_2v_2}}{b_{v_1v_2}} = \frac{1}{\lambda_{v_2v_1}}, \quad
        \tilde{\lambda}_{v_2v_1} = \frac{b_{v_1v_1}}{b_{v_2v_1}} = \frac{1}{\lambda_{v_1v_2}}.
    \end{equation*}
    This implies that the observed vector $X = (X_1, X_2)$ can be written in at least two ways,
    \begin{equation*}
        \setlength{\arraycolsep}{6pt}
        \begin{bmatrix}
            1 & -\lambda_{v_2v_1}\\
            -\lambda_{v_1v2} & 1
        \end{bmatrix}\begin{bmatrix}
            \varepsilon_1\\
            \varepsilon_2
        \end{bmatrix}  ,\quad
        \begin{bmatrix}
            1 & -1/\lambda_{v_1v_2}\\
            -1/\lambda_{v_2v_1} & 1
        \end{bmatrix}\begin{bmatrix}
            (-1/\lambda_{v_1v_2})\varepsilon_2\\
            (-1/\lambda_{v_2v_1})\varepsilon_1
        \end{bmatrix},
    \end{equation*}
    that are both compatible with the graph $\GG_2$. In other words, the matrix $\Lambda$ is not identifiable.
    \begin{figure}[t]
    \centering
    \begin{tikzpicture}        
    
        \node[minimum size = 1cm] (O) at (-1,0) {$\GG_2:$};
        \node[draw, circle, minimum size = 1cm] (A) at (0,0) {$v_1$};
        \node[draw, circle, minimum size = 1cm] (B) at (2,0) {$v_2$};           
            
            \draw[->] (A) to[bend left = 30] (B);
            \draw[->] (B) to[bend left = 30] (A); 
    \begin{scope}[xshift = 5cm, xscale = 1.3]
        \node[minimum size = 1cm] (O) at (-0.75,0) {$\GG_k:$};
        \node[draw, circle, minimum size = 1cm] (A) at (0,0) {$v_1$};
        \node[draw, circle, minimum size = 1cm] (B) at (1,1) {$v_2$};
        \node[] (C) at (2,1) {$\cdots$};
        \node[draw, circle, minimum size = 1cm] (D) at (3,1) {$v_{k-1}$};
        \node[draw, circle, minimum size = 1cm] (E) at (4,0) {$v_k$};
            
            \draw[->] (A) to[bend left = 25] (B);
            \draw[->] (B) to[bend left = 0] (C);
            \draw[->] (C) to[bend left = 0] (D);
            \draw[->] (D) to[bend left = 25] (E);
            \draw[->] (E) to[bend left = 15] (A); 
    \end{scope}
    \end{tikzpicture}
    \caption{On the left, a 2-cycle. On the right, a $k$-cycle}
    \label{fig:two:cycle}
\end{figure}
\end{example}

\begin{lemma}
\label{lem:k:cycle}
    Let $\GG_k$ be the $k$-cycle depicted in \cref{fig:two:cycle}. Then $\GG_k$ is generically identifiable if and only if $k\geq3$.
\end{lemma}
    
\begin{proof}
    From \cref{ex:non:suff}, we already know that $\GG_2$ is not identifiable. Hence, it is only left to show that $\GG_k$ is identifiable for $k\geq 3$. Herein, we present a proof for the case $k = 3$, as for larger cases, a similar argument would hold.

    The matrix $A$ of \cref{eq:A:matrix}, will have the following shape:
    \begin{equation}
    \label{eq:a:matr:3:cycle}
        A = 
        \setlength{\arraycolsep}{6pt} 
        \begin{bmatrix}
            b_{v_1v_1}-\tilde{\lambda}_{v_3v_1}b_{v_3v_1} & b_{v_1v_2}-\tilde{\lambda}_{v_3v_1}b_{v_3v_2} & b_{v_1v_3}-\tilde{\lambda}_{v_3v_1}b_{v_3v_3}\\[4pt]
            b_{v_2v_1}-\tilde{\lambda}_{v_1v_2}b_{v_1v_1} & b_{v_2v_2}-\tilde{\lambda}_{v_1v_2}b_{v_1v_2} & b_{v_2v_3}-\tilde{\lambda}_{v_1v_2}b_{v_1v_3}\\[4pt]
            b_{v_3v_1}-\tilde{\lambda}_{v_2v_3}b_{v_2v_1} & b_{v_3v_2}-\tilde{\lambda}_{v_2v_3}b_{v_2v_2} & b_{v_3v_3}-\tilde{\lambda}_{v_2v_3}b_{v_2v_3}\\[4pt]
        \end{bmatrix}.
    \end{equation}
    From \cref{ass:error} and the fact that there are no bidirected edges in the graph, we have
    \begin{equation}
    \label{eq:a:system:cycle}
        \begin{cases}
            & a_{v_1v_1}\cdot a_{v_2v_1} = a_{v_1v_1}\cdot a_{v_3v_1} = 0 \\
            & a_{v_1v_2}\cdot a_{v_2v_2} = a_{v_3v_2}\cdot a_{v_2v_2} = 0 \\
            & a_{v_1v_3}\cdot a_{v_3v_3} = a_{v_2v_3}\cdot a_{v_3v_3} = 0 \\
            & a_{v_1v_3}\cdot a_{v_2v_3} = a_{v_1v_2}\cdot a_{v_3v_2} = a_{v_2v_1}\cdot a_{v_3v_1} = 0.
        \end{cases}
    \end{equation}    
    If all the non-diagonal entries of $A$ are set to zero, we find $\tilde{\Lambda} = \Lambda$. We now show that this is the only solution for the system. Assume $a_{v_2v_1}\neq 0$, this implies $a_{v_1v_1} = 0$. This leads to 
    $$
    \tilde{\lambda}_{v_3v_1} = \frac{b_{v_1v_1}}{b_{v_3v_1}},
    $$
    plugging this value in \cref{eq:a:matr:3:cycle}, we obtain 
    \begin{equation}
    \label{eq:a:matr:3:cycle:2}
        A = \setlength{\arraycolsep}{6pt}
            \begin{bmatrix}
            0 & b_{v_1v_2}-({b_{v_1v_1}}/{b_{v_3v_1}})b_{v_3v_2} & b_{v_1v_3}-({b_{v_1v_1}}/{b_{v_3v_1}})b_{v_3v_3}\\[4pt]
            b_{v_2v_1}-\tilde{\lambda}_{v_1v_2}b_{v_1v_1} & b_{v_2v_2}-\tilde{\lambda}_{v_1v_2}b_{v_1v_2} & b_{v_2v_3}-\tilde{\lambda}_{v_1v_2}b_{v_1v_3}\\[4pt]
            b_{v_3v_1}-\tilde{\lambda}_{v_2v_3}b_{v_2v_1} & b_{v_3v_2}-\tilde{\lambda}_{v_2v_3}b_{v_2v_2} & b_{v_3v_3}-\tilde{\lambda}_{v_2v_3}b_{v_2v_3}\\[4pt]
        \end{bmatrix}.
    \end{equation}
    Writing explicitly the first row of the matrix in \cref{eq:a:matr:3:cycle:2}, one can see that 
    \begin{equation*}
        a_{v_1v_2} = \frac{\det\big((B_\Lambda)_{\{1, 3\}, \{1, 2\}}\big)}{b_{v_3v_2}}, \quad a_{v_1v_3} = \frac{\det\big((B_\Lambda)_{\{1, 3\}, \{1, 3\}}\big)}{b_{v_3v_3}},
    \end{equation*} 
    and both these quantities are different from zero for a generic choice of parameters in $\R^{\GG_3}_{\mathrm{reg}}$, see \cref{lem:gvl}. Having $a_{v_1v_2}\neq 0$, \cref{eq:a:system:cycle} implies that $a_{v_2v_2}=0$ and following the same argument as above, we conclude that $a_{v_2v_3}\neq0$. 
    Finally, we have $a_{v_1v_3}\cdot a_{v_2v_3} \neq 0$ since both terms are non-zero, which contradicts \cref{eq:a:system:cycle}. This proves that for a generic choice of parameters, the only solution of \cref{eq:a:system:cycle} is given by $\tilde{\Lambda} = \Lambda $. In other words, the matrix $\Lambda$ is generically identifiable.
\end{proof}

\begin{remark}
    It is noteworthy that \citet[Lemma 9]{drton:2011} proves that $k$-cycles with $k\geq 3$ are not generically identifiable from the covariance matrix alone. 
\end{remark}

\begin{theorem}
\label{thm:cycle}
    Let $\GG = (V, \Ed, \Eb = \emptyset)$ be a directed graph, such that $V = C_1\Dot{\cup}\cdots \Dot{\cup}C_n$, with $C_i$ being a $k_i$-cycle, and $\pa(C_i)\subseteq\bigcup_{j = 0}^i C_j$. Then $\GG$ is generically identifiable if and only if for every cycle $C = \{v_1, v_2\}$ of size $2$, we have $$\pa(C)\setminus C = \pa(v_1)\setminus C = \pa(v_2)\setminus C.$$
\end{theorem}
\begin{proof}
    \citet[\S2.3]{tramontano:2024:supp}.
\end{proof}
\section{Computational Experiments}
\label{sec:comp}
\subsection{Certifying Identifiability}
\label{subsec:cert}
We implemented the criterion from \cref{thm:cert:glob} using the algorithm of \cite{dinitz:1970} to solve the maximum flow problem. It operates with a complexity of $\mathcal{O}(|V|^4)$. Consequently, the algorithm we implemented has a complexity of $\mathcal{O}(|V|^5)$. 
We then determine the proportion of identifiable randomly sampled ADMGs of size $p = 25, 50$ and varying edge density. For each setup, we randomly sampled $5000$ graphs. More details on how the graphs were generated are provided in \citet[\S4.1.1]{tramontano:2024:supp}.

The proportions are displayed in Fig.~\ref{fig:identifiable:admg}. We observe that for the given sampling scheme, most graphs with an edge density below $0.4$ yield identifiable models, while for denser graphs, it becomes harder to find identifiable models. This follows the intuition that a very dense graph will have more confounders, so finding identifiable models is harder. The proportion of identifiable models remains similar across the two considered dimensions.

\begin{figure}[ht]
    \centering
    \includegraphics[scale = 0.45]{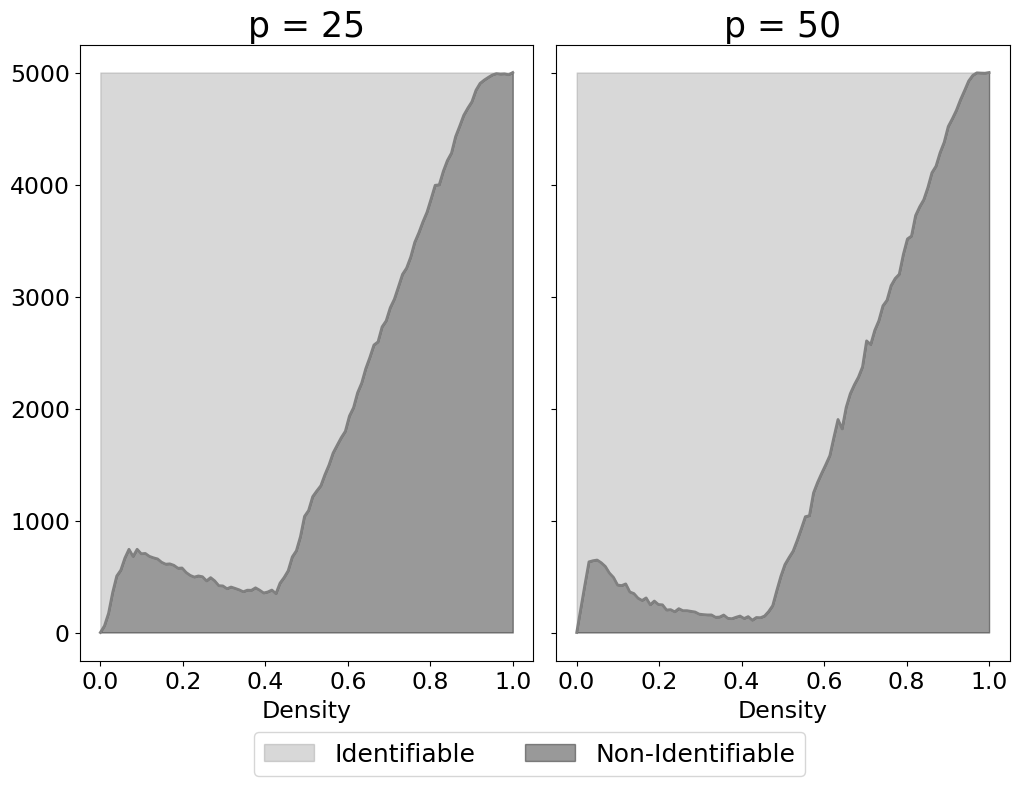}
    \caption{Proportion ADMGs for which every entry of the matrix $\Lambda$ is generically identifiable.}
    \label{fig:identifiable:admg}
\end{figure}

\subsection{Causal Effect Estimation}
\label{subsec:est}
Herein, we present an optimization problem that can be used to infer the identifiable causal effects.
\begin{lemma}
\label{lem:estimation}
Let $X=\Phi_\GG(\Lambda,\varepsilon)\in\MM(\GG)$, for a generic choice of $(\Lambda,\varepsilon)$, then $\Tilde{\Lambda}\in\R^{\GG_D}$ is a solution of \cref{eq:lambda:system} if and only if it is a solution of the optimization problem
\begin{equation}
    \label{eq:opt}
    \min_{\Tilde{\Lambda}\in\R^{\GG_D}}\left\{\sum_{u\xleftrightarrow{}v\notin\GG}\mu\left([(I-\tilde{\Lambda})\cdot X]_{u},[(I-\tilde{\Lambda})\cdot X]_{v}\right)\right\},
\end{equation}
where $\mu(\cdot,\cdot)$ is any consistent measure of dependence, i.e., any nonnegative real-valued function that takes as input two random variables and returns zero if and only if the random variables are independent.
\end{lemma}
\begin{proof}
See the supplemental material in \citet[\S2.4]{tramontano:2024:supp}
\end{proof}

For practical estimation, we may form an empirical version of the problem in \cref{eq:opt} by replacing the dependence measure $\mu$ with suitable consistent estimates.
One natural choice for $\mu$ is mutual information \citep[\S8.6]{cover:2006}. However, the most popular estimator for the mutual information is based on a k-nearest neighbor clustering of the sample, which would result in a non-smooth optimization problem \citep{kraskov:2003}. Several alternatives to mutual information have been proposed in the literature \citep{szekely:2007, Geenens:2022,shi:2022}. In particular, the Hilbert-Schmidt information criterion (HSIC) \citep{gretton:2007} has been extensively applied in causal inference \citep{mooij:2009, saengkyongam:2022}. For our empirical study, we used the HSIC, but other measures of independence can also be implemented.

When the underlying kernel is characteristic, the HSIC provides a measure of dependence that vanishes if and only if the variables for which it is computed are independent \citep[\S2.2 and Thm.~1]{fukumizu:2007}. Moreover, a consistent estimator for the HSIC \citep{gretton:2007} is given by  
$$
\widehat{\HSIC}_n(X, Y) := \trace(K_XHK_YH)/n^2,
$$
where $H_{i,j} = \delta_{i,j} - 1/n$ and $K_X$, and $K_Y$ denote the respective Gram matrices.

For a fixed graph $\GG$, and a given sample matrix $X\in \mathbb{R}^{n\times p}$, we estimate $\Lambda$ as a solution to the following optimization problem
\begin{equation}
\label{eq:sample:obj}
    \min_{\Tilde{\Lambda}\in\R^{\GG_D}}\left\{\sum_{u\xleftrightarrow{}v\notin\GG}\widehat{\HSIC}_n\left([(I-\tilde{\Lambda})\cdot X]_{u},[(I-\tilde{\Lambda})\cdot X]_{v}\right)\right\}.
\end{equation}
We used the L-BFGS method \citep{liu:1989} for solving the above optimization problem.
We considered two types of kernels in our experiments: radial basis function (RBF) kernels, the results of which are presented in \citet[Fig.~1]{tramontano:2024:supp}, and polynomial kernels, the results of which are depicted in \cref{fig:lap}. \citet[\S4.1.2]{tramontano:2024:supp} contains details on the data generation and additional experiments. It is noteworthy that in our experiments, the polynomial kernels of degree 2 \citep[\S2.3]{scholkopf:2018} provide a better estimate when initialized at the regression coefficient, compared to the RBF kernels for which the results rely more on the initialization.

In \cref{fig:lap}, we report the performance of our method on the IV graph (\cref{fig:IV}), the ADMG shown in \cref{fig:dc}, and the 3-cycle (\cref{fig:two:cycle}).
We used the normalized Frobenious loss between the estimated matrix $\hat{\Lambda}$ and the true matrix $\Lambda$, i.e., $||\hat{\Lambda} - \Lambda||_F/||\Lambda||_F$, as our loss function and reported the mean loss over fifty randomly sampled $\Lambda$.  
We compared our method against the Empirical Likelihood (EL) estimator proposed in \cite{wang:2017:el}. 
Note that for the IV graph, the parameters are identifiable from the covariance matrix. Therefore, the EL estimator, which is a covariance-based method, outperforms our method as the covariance matrix estimator is more sample-efficient than the HSIC estimator.
In contrast, for the ADMG in \cref{fig:dc} and the 3-cycle, the performance of the EL estimator does not improve with the sample size. This is due to the fact that the parameters of these two mixed graphs cannot be determined solely from the covariance matrix; see \citet[Prop.~2]{foygel:2012} and \citet[Lemma 9]{drton:2011}. When initialized at the regression coefficient, the performance of our proposed estimator improves with the sample size.
This indicates that the potential numerical issues arising from the non-convexity of the objective function and the estimation errors become less relevant as the sample size grows.
\begin{figure}
    \centering
    \includegraphics[scale=0.28]{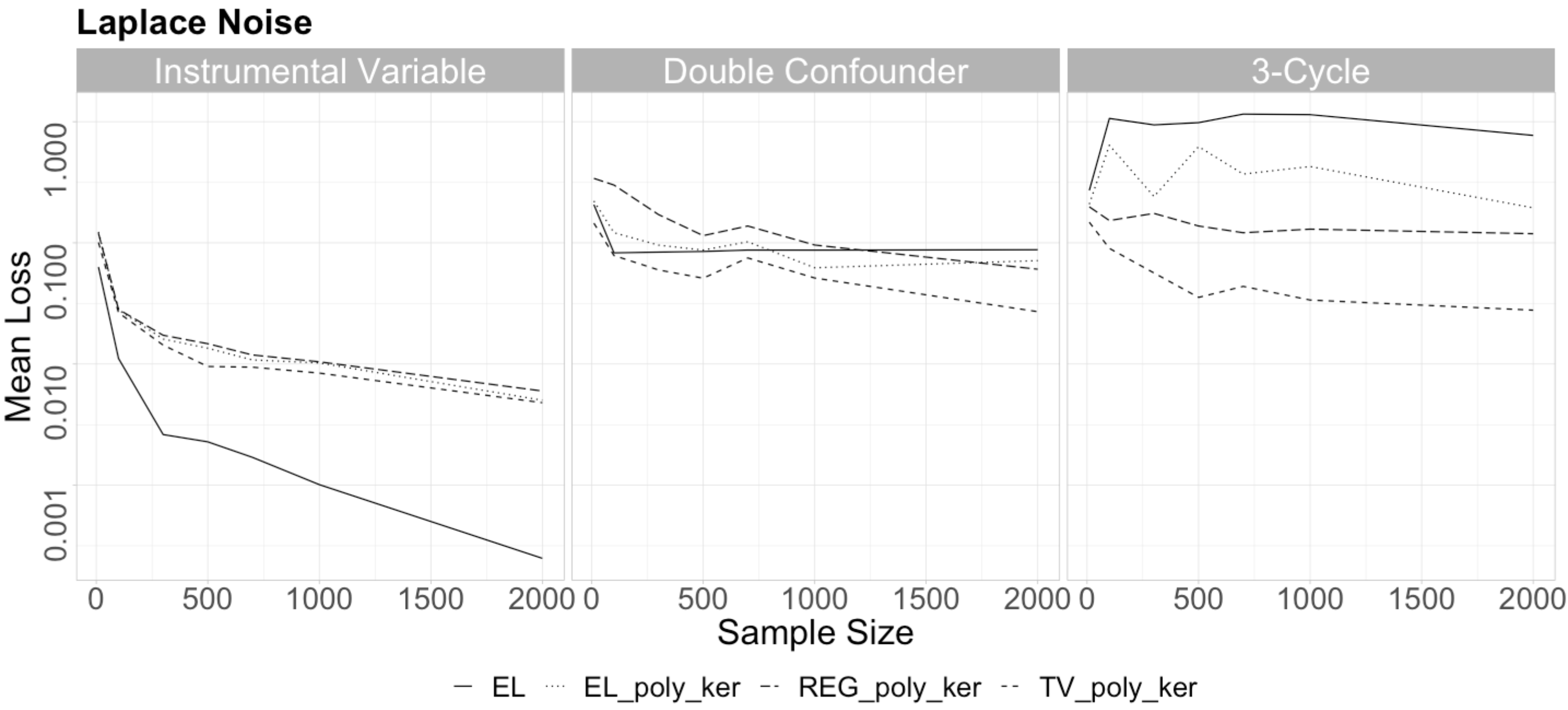}
    \caption{The performance of our proposed estimator with a polynomial kernel of degree 2 for different initializations on the IV graph, the ADMG of \cref{fig:dc}, and the 3-cycle. ``EL,''``REG,'' and ``TV'' represent the Empirical Likelihood, the regression coefficient, and the true value, respectively. Note that the $y$-axis is on a log scale.}
    \label{fig:lap}
\end{figure}

With the next theorem, we provide asymptotic consistency guarantees for the estimation method we used in the experiments.
\begin{theorem}
    \label{thm:consistency}
    Let $X = \Phi_\GG(\Lambda,\varepsilon) \in \MM(\GG)$. Consider an estimator $\Tilde{\Lambda}_n$ with
    \begin{equation}
    \label{eq:fs:optim}
        \Tilde{\Lambda}_n \in \argmin_{\tilde{\Lambda} \in \Theta} \left\{ \sum_{u \xleftrightarrow{} v \notin \GG} \widehat{\mathrm{HSIC}}_n \big([(I - \tilde{\Lambda}) \cdot X]_u, [(I - \tilde{\Lambda}) \cdot X]_v; \mathrm{k}_u, \mathrm{k}_u \big) \right\},
    \end{equation}
    where:
    \begin{itemize}
        \item $\Theta$ is a compact subset of $\R^{\GG_ D}$ containing $\Lambda$,
        \item $\mathrm{k}_u$ and $\mathrm{k}_v$ are fixed, bounded, continuously differentiable, characteristic kernels with bounded derivatives,
        \item $\varepsilon \in \mathcal{M}(\GG_B)$ has compact support.
    \end{itemize}
    Then, for a generic choice of $(\Lambda, \varepsilon) \in \R^{\GG_D} \times \mathcal{M}(\GG_B)$, we have 
    $$
    \lim_{n\to\infty}\mathbb{P}\left(\min_{\Lambda_0\in\Theta\cap\mathrm{Proj}}||\tilde{\Lambda}_n - \Lambda_0|| > \epsilon\right) = 0, \qquad \forall\epsilon > 0,
    $$    
    where $\mathrm{Proj} = \mathrm{Proj}_{\R^\GG}(\Phi_\GG^{-1}(X))$. In particular, if $\lambda_{uv}$ is generically identifiable according to \cref{thm:main:loc}, then $\tilde{\lambda}_{(n, uv)}$ converges to $\lambda_{uv}$ in probability.
\end{theorem}
\begin{proof}
See the supplemental material in \citet[\S2.4]{tramontano:2024:supp}
\end{proof}
\begin{remark}
    Fig.~5 in \citet[\S4.2]{tramontano:2024:supp} presents an example where only one of the two parameters is identifiable. Consistent with our theoretical predictions, only the estimator corresponding to the identifiable parameter demonstrates convergence as the sample size increases, while the estimator of the non-identifiable parameter fails to improve.
\end{remark}
\begin{remark}
The above result focuses on the finite-sample estimator of HSIC, as it is the one used in our simulations. However, a similar proof applies to any dependence measure for which pointwise convergence and Lipschitz continuity of the estimator can be guaranteed.
\end{remark}
\section{Conclusions}
\label{sec:conc}
In this work, we studied the generic identifiability of direct causal effects in linear structural equation models with dependent errors. For acyclic models, we obtained a complete graphical characterization of the identifiable causal effects, with a graphical criterion that can be checked in polynomial time in the size of the graph. For cyclic models, we proved that the same graphical conditions are necessary for identifiability, and we provided counterexamples to show that they are not sufficient. For a smaller family of cyclic models, we provided a complete graphical characterization of the identifiable effects. However, a complete characterization of the identifiability for cyclic models involves additional mathematical subtleties and is left as a problem for future work. 

Most of the literature on identifiability in linear structural equation models leverages specific moment equations to obtain identifiability results. In this work, we follow a different approach and exploit the information contained in the whole distribution, explicitly leveraging the independence relations dictated by the missing bidirected edges in the graph. To the best of our knowledge, our work is the first to follow this route in this generality.  In an initial exploration of parameter estimation, we showed that estimates obtained by minimizing structurally absent dependences can be useful.

\medskip

To conclude, we highlight possible future directions.
\paragraph*{Beyond Observational Data} In this paper, we considered the situation when only observational data is available. Recent identification results that additionally consider information from interventional datasets have been proposed for non-parametric models \citep{lee:2020, kivva2022revisiting}. Extending our results to these setups can be seen as a natural future direction.
\paragraph*{Non-linear Models} In the graphical models literature, different non-parametric assumptions on the functional relations among the variables have been used to guarantee the identifiability of the causal structure \cite[\S7.1]{Peters:Elements:2017}. Similar assumptions have been used to prove the identifiability of the causal effect under specific causal assumptions, e.g., \citet{imbens:2009}. However, a general graphical criterion for identification in non-linear structural equation models is currently missing. We believe that the ideas we propose in this work admit suitable extensions to these more general settings.
\paragraph*{Confidence Intervals} In this paper, we have focused exclusively on point estimation of causal effects. The construction of confidence intervals for these effects is a promising direction for future research. Both test inversion methods \citep[Section~3.4]{saengkyongam:2022} and the score test approach of \cite{londschien:2024} could be potentially be adapted to obtain confidence intervals in our context. Importantly, confidence intervals would also provide a practical means to test the genericity assumptions since violations of these assumptions would typically manifest as unusually wide confidence intervals.
\paragraph*{Relaxing Independence} 
Throughout this paper, we have assumed that missing edges in the bidirected part of the graph correspond to independence constraints on the exogenous noise vector. In light of recent developments in \emph{non-independent} component analysis \citep{mesters:2022,garrotelopez:2024,ribot:2025}, it would be natural to investigate whether our graphical criteria extend to settings where weaker moment conditions are assumed in lieu of independence.
\paragraph*{Structure Identifiability} 
In this paper, we have studied a parameter identifiability problem, assuming that the mixed graph $\GG$ is given. However, there are numerous practical problems in which one needs to infer the graph from data. This model selection problem is often referred to as \emph{structure learning} or \emph{causal discovery} \citep{drton:maathuis:2017}. Therefore, understanding whether, within the model class under consideration, the graph itself is identifiable can be seen as a natural future direction.


\begin{funding}
This project has received funding from the European Research Council (ERC) under the European Union’s Horizon 2020 research and innovation programme (grant agreement No 883818). DT's PhD scholarship is funded by the IGSSE/TUM-GS via a Technical University of Munich--Imperial College London Joint Academy of Doctoral Studies.
\end{funding}

\begin{supplement}
\stitle{Supplement to ``Parameter identification in linear
non-Gaussian causal models under general confounding''.}
\sdescription{The supplement contains further preliminary material, details of the proofs, and additional numerical experiments.}
\end{supplement}
\begin{supplement}
\stitle{Code for ``Parameter identification in linear
non-Gaussian causal models under general confounding''.}
\sdescription{The folder contains all the code necessary to implement the algorithms introduced in the paper and to replicate all the experiments presented.}
\end{supplement}


\clearpage

\begin{center}
    \bf\large \sc Supplement to ``Parameter identification in linear
non-Gaussian causal models under general confounding''
\end{center}

\section*{}
In \cref{app:notions} of this supplementary article, we provide additional definitions and known results that we will need for the proofs. \cref{app:sec:proofs} contains the proofs for the results of the main paper. \cref{app:oica} provides a detailed comparison with identifiability results obtained through Overcomplete ICA. \cref{app:simulations} contains the details of the experimental setting, as well as additional experiments.
\section{Notions of Non-Linear Algebra}
\label{app:notions}
In this section, we give the basic definitions of \emph{non-linear} algebra we will need for the proofs; we defer the interested reader to \cite{cox:2015,michalek:2021} for more details.
\begin{definition}
\label{def:ag:basics}
    For every natural number $n$, we denote the ring of polynomials in $n$ variables $x_1,\dots,x_n$ by $\R[x_1,\dots,x_n]$. Let $S$ be a possibly infinite subset of $\R[x_1,\dots,x_n]$. The affine variety associated to $S$ is defined as  $\mathcal{V}(S)=\{x\in\R^n\::\: f(x)=0,\,\forall f\in S\}$. The vanishing ideal associated to a variety $\mathcal{V}\subseteq\R^n$ is $\mathcal{I}(\mathcal{V})=\{f\in\R[x_1,\dots,x_n]\::\: f(x)=0.\,\forall x\in\mathcal{V}\}$, and the coordinate ring of $\mathcal{V}$ is defined as $\R[\mathcal{V}]=\R[x_1,\dots,x_n]/\mathcal{I}(\mathcal{V})$.
\end{definition}

\begin{lemma}{\citet[Prop.~3.1]{sullivant:2010}}
\label{lem:B:matrix}
Let $\GG_D$ be any directed graph. For every $\Lambda\in\R^{\GG_D}_\mathrm{reg}$ we have
        \begin{equation*}
        \label{eq:B:matrix}
            (B_{\Lambda})_{uv}=(I-\Lambda)^{-T}_{uv}=\sum_{P\in\mathcal{P}(v,u)}\lambda^P,
        \end{equation*}
        in particular $(B_{\Lambda})_{uv}=0$
        if $u\notin\de(v)$.
\end{lemma}
    
\section{Proofs for the Main Paper}
\label{app:sec:proofs}
\subsection{Proofs for Section 4 from the Main Paper}
\begin{proof}[Proof for Lemma 4.1 from the main paper]
    First, notice that for all practical purposes, we can consider the edge capacity to be $|\pa(v)|$ instead of $\infty$; this implies that we can exploit additional properties of maximum flow problems with integer values.
    
    We are going to show that to every flow $f$ in $G^v_Q$ of size $k$ with integer values, we can associate $(I_f,P_f)\in2^{R_v}\times2^Q$ and a system of directed paths $\Pi_f\in\tilde{\PP}(I_f,P_f)$ such that $|I_f|=k$, and vice versa. That is, for every pair $(I,P)\in2^{R_v}\times2^Q$ and a system of directed paths $\Pi=(\pi_1,\dots,\pi_k)\in\tilde{\PP}(I,P)$ such that $|I|=k$, we can associate an integer flow $f_{I,P}$ of size $k$.
    
    Let us first consider a flow $f$ with integer value in $G^v_Q$. Since the capacity of each node, that is not a sink or a source, is $1$, we can restrict the image of $f$ to be $\{0,1\}$.
    Define $I_f=\{u\in R_v\::\: f(s_v,u)=1\}$.  Since the size of the flow is $k$, we have $|I_f|=k$. For every $u\in I_f$ consider the directed path $\pi_u: u=:u_0,u_1,\dots,u_{k_u}$ such that $f(u_i,u_{i+1})=1$ and $u_{k_u}\in Q$.  This is well defined since for every $u_i\in V_v\setminus\{s_v,t_v\}$ there is at most one other $u_{i+1}\in V_v\setminus\{s_v,t_v\}$ such that $f(u_i,u_{i+1})=1$, and if one assumes that there is an $u_{i-1}$ such that $f(u_{i-1},u_{i})=1$ then the existence of $u_{i+1}$ is guaranteed from the first equality in Eq.~4.1 from the main paper.
    Let $P_f=\{u_{k_u}\::\: u\in I_f\}$.  
    To obtain a contradiction, we prove that $\Pi=(\pi_u\::\: u\in I_v)\in\Tilde{\PP}(I_f,P_f)$. 
    Suppose two directed paths in $\Pi$ intersect at a node $u$. This 
 implies that there are $u_0\neq u_1\in V_v$ such that $f(u_0, u)=f(u_1, u)=1$, hence $\sum_{w\in V_v}f(w,u)\geq 2>c_V(u)$. We obtain a violation of Eq.~4.1 in the main paper.
    
    For the other implication, consider $(I,P)\in2^{R_v}\times2^Q$ and a system of directed paths $\Pi=(\pi_1,\dots,\pi_k)\in\tilde{\PP}(I,P)$ such that $|I|=k$. We define $f_{I,P}$ as follows:
    \begin{equation}
    \label{eq:ass:flow}
        f_{I,P}(u,w)=\begin{cases}
        1, \qquad&\emph{if } u=s_v \emph{ and } w\in I \emph{ or } u\in P \emph{ and } v=t_v, \\
        1, \qquad&\exists \,j\in[k]\::\: u\to v\in\pi_j,\\
        0, \qquad&\emph{otherwise}.
        \end{cases}
    \end{equation}
    We need to show that $f_{I, P}$ satisfies Eq.~4.1 from the main paper. Since the capacity of each edge is infinity, we only need to check the first inequality; this holds because $\Pi$ is a \emph{non-intersecting} system of directed paths, and so each node has at most one incoming, outgoing, for which the flow is different from 0. 
     By directly plugging in \cref{eq:ass:flow} into Eq.~4.2 from the main paper, one can see that $|f_{I, P}|=k$.
    To conclude the proof, we need to show that there is a solution to $G^v_Q$ with integer values. This is ensured by applying \citet[Thm.~26.10]{cormen:2009} and the fact that all the capacities in $G^v_Q$ are integers. 
\end{proof}

\begin{proof}[Proof for Theorem 4.2 from the main paper]
    \cite{chen:2022} proved that the complexity of any maximum flow problem $(G,s,t,c_V,c_D)$ is \emph{almost linear} in the number of edges in the graph $G$. For every node $v\in V$ and $Q\subseteq\pa(v)$ to certify the identifiability of $\lambda_{Q,v}$, one needs to solve the maximum flow problems $G^v_Q$ and $G^v_{\pa(v)}$ and then check whether the difference of the sizes of the corresponding maximum flows is $|Q|$. Since both $G^v_Q$ and $G^v_{\pa(v)}$ have at most $2(|V|+1)$ nodes, the overall complexity is $\mathcal{O}(|V|^{2+\mathrm{Obj}(1)})$.
\end{proof}

\begin{proof}[Proof for Theorem 4.3 from the main paper]
    To certify the identifiability of all the directed edges, i.e., the whole matrix, one needs to solve the maximum flow problem $G^v_{\pa(v)}$ for every $v$ in $V$ and check whether the maximum flow has the size $|\pa(v)|$. This adds a multiplicative factor $|V|$ to the result of Theorem 4.2 from the main paper, which leads to  $\mathcal{O}(|V|^{3+\mathrm{Obj}(1)})$.
\end{proof}
\subsection{Proofs for Section 5.2 from the main paper}
\label{app:subsec:proof:gen:ass}

In the sequel, we will consider $\mathcal{M}^{\leq k}(\GG_B)$ as a variety in the space given by the Cartesian product of the symmetric tensor spaces $(\Sym_l(\mathbb{R}^p))_{2\le l\le k}$, which is isomorphic to $\R^{\sum_{s=2}^k \binom{p+s-1}{s}}$. We denote the corresponding coordinate ring as $\R[\mathcal{M}^{\leq k}(\GG_B)]$. For every $k$-tuple $(i_1,\dots,i_k)$, we denote by $\mathcal{M}^{(k)}_{\setminus (i,\dots,i)}(\GG_B)$ the projection of $\mathcal{M}^{(k)}(\GG_B)$ on the coordinates not corresponding to the entry $(i_1,\dots,i_k)$.

\begin{proof}[Proof for Lemma 5.5 from the Main Paper]
The fact that $\phi^k$ is well defined is a consequence of \citet[Prop.~3.1]{comon:jutten:handbook}. 

There is a one-to-one linear transformation between cumulants and moments; see, e.g., \citet[\S2.3]{mccullagh:1987}; hence, it is enough to prove the result for the corresponding set of moments. It is known that the set of symmetric tensors that can be generated as a moment of a distribution is a full dimensional convex cone in the space of symmetric tensors, see, e.g., \citet[Lem.~3.3]{didio:2022}. Hence, the same result holds for the set of cumulants, $\phi^{\leq k}(\mathcal{M}_\infty(\GG_B))$ is the projection of this convex cone along the coordinate axes corresponding to connected subsets of $\GG_B$, so is itself a full dimensional convex cone in $\mathcal{M}^{\leq k}(\GG_B)$.
\end{proof}

\begin{lemma}[\cite{comon:jutten:handbook}, \S5, Eq.~5.8]
\label{lem:multilin:cum}
    Let $\varepsilon=(\varepsilon_1,\dots,\varepsilon_p)$ be any $p$-variate random vector, and $A\in\mathbb{R}^{s\times p}$ for any $s\in\mathbb{N}$, then \begin{equation*}
        \begin{aligned}
            \mathcal{C}^{(k)}(A\cdot\varepsilon)_{i_1,\dots,i_k} 
             = \sum_{{1\leq}j_1,\dots,j_k{\leq p}}\mathcal{C}^{(k)}(\varepsilon)_{j_1,\dots,j_k}a_{j_1i_i}\cdots a_{j_ki_k}.
        \end{aligned}
    \end{equation*}
\end{lemma}

\begin{lemma}
\label{lem:transf:cumulant}
Let $\varepsilon\in\mathcal{M}_{\infty}(\GG_B)$, and $A\in\mathbb{R}^{2\times p}$, then 
\begin{equation}
    \label{eq:poly:par:cum}
    \mathcal{C}^{(k)}(A\cdot\varepsilon)_{i_1,\dots,i_k}=\sum_{\substack{\{j_1,\dots,j_k\} \emph{ is }\\\emph{connected in } \GG_B} }\mathcal{C}^{(k)}(\varepsilon)_{j_1,\dots,j_k}a_{i_ij_1}\cdots a_{i_kj_j}{.}
\end{equation}
{In particular, $\mathcal{C}^{(k)}(A\cdot\varepsilon)_{i_1,\dots,i_k}$ is a \emph{non-zero} polynomial in $\R[\mathcal{M}^{(k)}(\GG_B), a_{i,j} \::\: i,j\in[p]]$.}
\end{lemma}

\begin{proof}
A direct consequence of \cref{lem:multilin:cum} above and Lemma 5.5 from the main paper. The polynomial is \emph{non-zero} because each monomial of the form $\mathcal{C}^{(k)}(\varepsilon)_{j_1,\dots,j_k}$ only appears once in the summation, hence it cannot cancel out with any other monomial.
\end{proof}

\begin{proof}[Proof for Theorem 5.6 from the main paper]
From Lemma 5.5 in the main paper, we know that $\dim(\phi^{\leq k}(\mathcal{M}_\infty(\GG_B)))=\dim(\mathcal{M}^{\leq k}(\GG_B))$. Hence, it is enough to show that $\phi^{\leq k}(\mathcal{S}(\GG_B))$ lies in a subvariety of $\mathcal{M}^{\leq k}(\GG_B)$ of strictly smaller dimension, see e.g., \citet[Lemma]{okamoto:1973}.

Notice that we can write $$\mathcal{S}(\GG_B)=\left(\bigcup_{i\in[p]}\mathcal{S}_{i}(\GG_B)\right)\cup\left(\bigcup_{u_i\leftrightarrow{}u_j\in\GG_B}\mathcal{S}_{i\leftrightarrow{}j}(\GG_B)\right),$$
where

\begin{equation*}
    \begin{aligned}
            \kappa_{i}(\varepsilon)=\{A\in\kappa(\varepsilon)\::\: a_{1i}\cdot a_{2i}\neq0\},\qquad \mathcal{S}_{i}(\GG_B)=\{\varepsilon\in\mathcal{S}(\GG_B)\::\:\kappa_{i}(\varepsilon)\neq\emptyset\},
    \end{aligned}
\end{equation*}
while $\kappa_{i\leftrightarrow{}j}(\varepsilon), \mathcal{S}_{i\leftrightarrow{}j}(\GG_B)$ are defined in a similar way. Hence, it is enough to prove that both $\phi^{\leq k}(\mathcal{S}_{i}(\GG_B))$ and $\phi^{\leq k}(\mathcal{S}_{i\leftrightarrow{}j}(\GG_B))$ are Lebesgue measure 0 subsets of $\mathcal{M}^{\leq k}(\GG_B)$ for $k$ high enough. 

We start with by bounding the dimension of $\mathcal{S}_{i}(\GG_B)$. For every $\varepsilon\in\mathcal{S}_{i}(\GG_B)$, every $A\in\kappa_i(\varepsilon)$, and every $0\neq s, t\in\mathbb{N}$ we can use \cref{lem:transf:cumulant} to write 
$$
0=\mathcal{C}^{(s+t)}(A\cdot\varepsilon)_{\underbrace{1,\dots,1}_{s},\underbrace{2\dots,2}_{t}}=\mathcal{C}^{(s+t)}(\varepsilon)_{i,\dots,i}a_{1i}^s\cdot a^t_{2i}+r^{s+t}_{\setminus i}(\varepsilon, A)
$$ 
where $r^{s+t}_{\setminus i}(\varepsilon, A)$ is a \emph{non-zero} polynomial in $\R[\mathcal{M}_{_{\setminus (i,\dots,i)}}^{(s+t)}(\GG_B), a_{i,j} \::\: i,j\in[p]]$, notice that for the first equality we used that $(A\varepsilon)_1$ and $(A\varepsilon)_2$ are independent{, while for the second equality we isolated the element corresponding to $\{i,\dots,i\}$ in \cref{eq:poly:par:cum}}. This implies that we can write 
\begin{equation}
\label{eq:image:incl}
    \mathcal{C}^{(s+t)}(\varepsilon)_{i,\dots,i}=\phi^{k}(\varepsilon)_{(i,\dots,i)}=-{r^{s+t}_{\setminus i}\left(\phi^{k}(\varepsilon)_{\setminus (i,\dots,i)}, A\right)}/{a_{i1}^s\cdot a^t_{i2}}.
\end{equation}
 We can define a rational map $\psi^{s,t}_{i}: \mathcal{M}_{_{\setminus (i,\dots,i)}}^{(s+t)}(\GG_B)\times \R^{2p}\to\mathcal{M}^{(s+t)}(\GG_B)$ in the following way 
$$
    \psi^{s,t}_i( \mathcal{C}^{s+t},{A}    
    )_{i_1,\dots,i_k}:=\begin{cases}
    -r^{s+t}_{\setminus i}(\mathcal{C}^{s+t}_{\setminus (i,\dots,i)}, A)/a_{1i}^s\cdot a^t_{2i}, &\emph{ if } (i_1,\dots,i_k)=(i,\dots,i),\\
    \mathcal{C}^{s+t}_{i_1,\dots,i_k}, &\emph{otherwise},
    \end{cases}
$$
see, e.g., \citet[\S5]{cox:2015} for the definition of rational map.

What \cref{eq:image:incl} shows is that $$\phi^{(s+t)}(\mathcal{S}_i(\GG_B))\subseteq \psi^{(s,t)}(\mathcal{M}_{_{\setminus (i,\dots,i)}}^{(s+t)}(\GG_B)\times \R^{2p})\subseteq\mathcal{M}^{(s+t)}(\GG_B).$$ Let us consider $\psi^{\leq k}_{i}:\mathcal{M}_{_{\setminus (i,\dots,i)}}^{\leq k}(\GG_B)\times \R^{2p}\to\mathcal{M}^{\leq k}(\GG_B)$ such that $\psi^{\leq k}_{i}\circ \pi_{k_0}=\psi^{k_0-1,1}_i$, where $\pi_{k_0}$ is the projection of $\mathcal{M}^{\leq k}$ onto $\mathcal{M}^{(k_0)}$ for every ${2\leq}k_0\leq k$. Again, we have $\phi^{\leq k}(\mathcal{S}_i(\GG_B))\subseteq \psi^{\leq k}(\mathcal{M}_{_{\setminus (i,\dots,i)}}^{\leq k}(\GG_B)\times \R^{2p})\subseteq\mathcal{M}^{\leq k}(\GG_B)$, that concludes the proof by noticing that 
\begin{equation*}
    \begin{aligned}
        &\dim(\phi^{\leq k}(\mathcal{S}_i(\GG_B)))\leq\dim\bigg(\psi^{\leq k}(\mathcal{M}_{_{\setminus (i,\dots,i)}}^{\leq k}(\GG_B)\times \R^{2p})\bigg) \\
        &\leq  \dim(\R^{2p})+\dim(\mathcal{M}_{_{\setminus (i,\dots,i)}}^{\leq k}(\GG_B))\leq 2p+\dim(\mathcal{M}^{\leq k}(\GG_B))-(k-1),
    \end{aligned}
\end{equation*}
that is strictly smaller than $\dim(\mathcal{M}^{\leq k}(\GG_B))$ if $k\geq 2(p+1)$. 

In order to prove the result for $\mathcal{S}_{i\leftrightarrow{} j}(\GG_B)$, we first notice that we can always write $$\mathcal{S}_{i\leftrightarrow{} j}(\GG_B)=\left(\mathcal{S}_{i\leftrightarrow{} j}(\GG_B)\cap\bigcup_{i\in[p]}\mathcal{S}_i(\GG_B)\right)\dot{\cup}\left(\mathcal{S}_{i\leftrightarrow{} j}(\GG_B)\setminus\bigcup_{i\in[p]}\mathcal{S}_i(\GG_B)\right).$$  Since we have already bounded the dimension of $\mathcal{S}_{i\leftrightarrow{} j}(\GG_B)\cap\bigcup_{i\in[p]}\mathcal{S}_i(\GG_B)$; to conclude the proof we only need to bound the dimension of  $$\tilde{\mathcal{S}}_{i\leftrightarrow{} j}(\GG_B):= \mathcal{S}_{i\leftrightarrow{} j}(\GG_B)\setminus\bigcup_{i\in[p]}\mathcal{S}_i(\GG_B).$$ 
For every $\varepsilon\in\tilde{\mathcal{S}}_{\leftrightarrow{} j}(\GG_B)$, and any $A\in\kappa_{\leftrightarrow{} j}(\varepsilon)$, and every $2\leq k\in\mathbb{N}$ we can use \cref{lem:transf:cumulant} to write 
$$
0=\mathcal{C}^{(k)}(A\cdot\varepsilon)_{1,\dots,1,2}=\mathcal{C}^{(s+t)}(\varepsilon)_{i,\dots,i,j}a_{1i}^{k-1}\cdot a_{2j}+r^{k}_{\setminus i\leftrightarrow{} j}(\varepsilon, A),
$$ 
where we used that $a_{1i}\cdot a_{2i}=0$ for every $i$, that is a consequence of $\varepsilon\notin\cup_{i\in[p]}\mathcal{S}_i(\GG_B)$ to simplify the formula given in \cref{lem:transf:cumulant}. This allows us to write 
\begin{equation*}
        \mathcal{C}^{(k)}(\varepsilon)_{i,\dots,i,j}=\phi^{k}(\varepsilon)_{(i,\dots,i,j)}=-{r^{k}_{\setminus i\leftrightarrow{} j}\left(\phi^{k}(\varepsilon)_{\setminus (1,\dots,1,2)}, A\right)}/{a_{1i}^{k-1}\cdot a_{2j}}.
    \end{equation*}
The rest of the proof follows verbatim the case of $\mathcal{S}_i(\GG_B)$.
\end{proof}

\begin{example}
\begin{figure}[h]
    \centering
    \begin{tikzpicture}
            \node[draw, circle, inner sep=3pt] (A) at (0,0) {$v_1$};
            \node[draw, circle, inner sep=3pt] (B) at (2,0) {$v_2$};
            \node[draw, circle, inner sep=3pt] (C) at (4,0) {$v_3$};
        
            \draw[dashed][<->] (A) to[bend left=30] (B);
        \end{tikzpicture}
        
    \caption{A bidirected graph with three nodes.}
    \label{fig:3:nodes:ex}
\end{figure}
In this example, we compute the map $\psi^{1, 1}_1$ introduced in the proof above for the graph depicted in \cref{fig:3:nodes:ex}. For any $A\in\R^{2\times 3}$, and $\varepsilon\in\MM_{\infty}(\GG_B)$ we have 
$$
\mathcal{C}^{(2)}(A\cdot\varepsilon)_{1, 2}= \mathcal{C}^{(2)}(\varepsilon)_{1, 1}a_{1, 1}a_{1, 2} + \mathcal{C}^{(2)}(\varepsilon)_{1, 2}a_{1, 1}a_{2, 2} + \mathcal{C}^{(2)}(\varepsilon)_{2, 2}a_{2, 1}a_{2, 2} + \mathcal{C}^{(2)}(\varepsilon)_{3, 3}a_{3, 1}a_{3, 2}.
$$
If $(A\varepsilon)_1\indep (A\varepsilon)_2$, then $\mathcal{C}^{(2)}(A\cdot\varepsilon)_{1, 2} = 0$, and thus
$$
\mathcal{C}^{(2)}(\varepsilon)_{1, 1} = -\frac{\mathcal{C}^{(2)}(\varepsilon)_{1, 2}a_{1, 1}a_{2, 2} + \mathcal{C}^{(2)}(\varepsilon)_{2, 2}a_{2, 1}a_{2, 2} + \mathcal{C}^{(2)}(\varepsilon)_{3, 3}a_{3, 1}a_{3, 2}}{a_{1, 1}a_{1, 2}} = -\frac{r_{\setminus 1}^{1, 1}(\mathcal{C}^{(2)}_{\setminus (1, 1), A})}{a_{1, 1}a_{1, 2}}.
$$
Notice that the above equation is well defined only if $a_{1, 1}a_{1, 2}\neq0$, which is true for the distributions in $\varepsilon\in\mathcal{S}_1(\GG_B)$.
The map $\psi_1^{1,1}$ is defined as follows:
\begin{equation*}
    \begin{aligned}
        \psi_1^{1,1}&\left(\mathcal{C}^{(2)}(\varepsilon)_{1, 2}, \mathcal{C}^{(2)}(\varepsilon)_{1, 3}, \mathcal{C}^{(2)}(\varepsilon)_{2, 2}, \mathcal{C}^{(2)}(\varepsilon)_{2, 3}, \mathcal{C}^{(2)}(\varepsilon)_{3,3}, a_{1, 1}\dots,a_{2, 3}\right) =\\
        &\left(-\frac{r_{\setminus 1}^{1, 1}(\mathcal{C}^{(2)}_{\setminus (1, 1), A})}{a_{1, 1}a_{1, 2}}, \mathcal{C}^{(2)}(\varepsilon)_{1, 3}, \mathcal{C}^{(2)}(\varepsilon)_{2, 2}, \mathcal{C}^{(2)}(\varepsilon)_{2, 3}, \mathcal{C}^{(2)}(\varepsilon)_{3,3}
        \right).
    \end{aligned}
\end{equation*}
\end{example}

\subsection{Proofs for Section 6 from the Main Paper}
\label{app:subsec:proof:cycle}
\begin{lemma}
\label{lem:diag:cycle}
    Let $\GG = (V, \Ed, \Ed)$ be a mixed graph.  Assume the vertex set can be partitioned as $V = C_1\Dot{\cup}\cdots \Dot{\cup}C_n$, with $C_i$ being a $k_i$-cycle, and $\pa(C_i)\subseteq\bigcup_{j = 0}^i C_j$, where $\Dot{\cup}$ denotes the union of disjoint sets. Then, $\GG$ is generically identifiable if and only if $\lambda_{C_i, v}$ is identifiable for every $i\in[n]$ and $v\in C_i$, and the graphical criterion in Theorem 6.3 from the main paper is satisfied.
\end{lemma}

\begin{proof}
    If the matrix $\Lambda$ is identifiable, then by definition, all of its columns are also identifiable, and from Theorem 6.3 from the main paper, we know that the graphical condition is satisfied. We now prove that the reverse implication is also true.

    By plugging in $\lambda_{C_i, v}$ instead of $\tilde{\lambda}_{C_i, v}$ in Eq.~3.1 from the main paper, one can see that the matrix $A$ has the following shape
    
    \begin{equation*}
    \setlength{\arraycolsep}{6pt} 
        \begin{bmatrix}
                I_{k_1} & 0 & \cdots & 0 \\[4pt]
                A_{C_2,C_1} & I_{k_2} & \cdots & 0 \\[4pt]
                \vdots & \vdots & \ddots & \vdots\\[4pt] 
                A_{C_n, C_1} & A_{C_n, C_2} & \cdots & I_{k_n} \\[4pt]
        \end{bmatrix}.
    \end{equation*}
    In particular, we have $a_{v,v} = 1$ for every $v\in V$. The same proof as in Lemma 3.3 from the main paper applies.
\end{proof}

\begin{proof}[Proof for Theorem 6.5 from the main paper]
    We know from Lemma 6.4 in the main paper that if $k_i\neq 2$ then $\lambda_{C_i, v}$ is identifiable. If the set $S = \{i\in [n]\::\: k_i = 2\}$ is empty then we know from \cref{lem:diag:cycle} and the fact that $\Eb = \emptyset$ that $\Lambda$ is identifiable. Otherwise, let $m = \min S$ and $C_m = \{v_1, v_2\}$.

    We know from Example 6.2 in the main paper that we can choose $(\tilde{\lambda}_{v_1v_2}, \tilde{\lambda}_{v_2v_1}) = (b_{v_2v_2}/b_{v_1v_2}, b_{v_1v_1}/b_{v_2v_1})$. If $m = 1$, letting $\tilde{\lambda}_{u, v} = \lambda_{u, v}$ for $v\notin C_1$, the matrix $A$ of Eq.~5.4 from the main paper will have the following shape
    \begin{equation*} 
        \setlength{\arraycolsep}{6pt}
        \begin{bmatrix}
                0 & \det((B_{\Lambda})_{\{v_1, v_2\}, \{v_1, v_2\}})/b_{v_2v_1} & 0 & \cdots & 0\\[4pt]
                \det((B_{\Lambda})_{\{v_1,v_2\}, \{v_1, v_2\}})/b_{v_1v_2} & 0 & 0 & \cdots & 0\\[4pt]                                
                0 & 0 & I_{k_2} & \cdots & 0 \\[4pt]
                \vdots & \vdots & \vdots & \ddots & \vdots\\[4pt]
                0 & 0 & 0 & \cdots & I_{k_n} \\[4pt]
        \end{bmatrix},
    \end{equation*}
    that satisfies all the constraints imposed by Assumption 1 from the main paper. Proving that $\Lambda$ is not identifiable in this case.

    If $m > 1$, we know that $\lambda_{C_i, v}$ is identifiable for every $i < m$, hence the matrix $A$ will be as follows
    \begin{equation*}
        \setlength{\arraycolsep}{6pt}
        \begin{bmatrix}
                I_{k_1} & 0 & \cdots & \cdots &\cdots & \cdots &  0\\[4pt]
                0 & I_{k_2} & \cdots & \cdots &\cdots & \cdots & 0\\[4pt]
                \vdots & \vdots & \ddots & \ddots & \ddots & \vdots & \vdots\\[4pt]
                A_{v_1, C_1} & A_{v_1, C_2} & \cdots & 0 & \det(B_{\Lambda})_{\{v_1, v_2\}, \{v_1, v_2\}}/b_{v_2v_1}  &\cdots & 0\\[4pt]
                A_{v_2, C_1} & A_{v_2, C_2} & \cdots & \det(B_{\Lambda})_{\{v_1,v_2\}, \{v_1, v_2\}}/b_{v_1v_2}  & 0 &\cdots & 0\\[4pt]
                \vdots & \vdots & \vdots & \ddots & \ddots & \ddots& \vdots\\[4pt]
                A_{C_n, C_1} & A_{C_n, C_2} & \cdots & \cdots &\cdots & \cdots & A_{C_n, C_n}\\[4pt]
        \end{bmatrix}.
    \end{equation*}
    This implies that in order for the matrix $A$ to satisfy the conditions in Assumption 1 from the main paper for every pair of nodes, we must have $A_{C_m, \an(C_m)\setminus C_m} = 0$. This might happen if and only if $A_{C_m, \pa^*(C_m)} = 0$, where $\pa^*(C_m) = \pa(C_m)\setminus C_m$. Writing $A_{v_1, \pa^*(C_m)} = 0$ explicitly we get to the following linear system
    \begin{equation}
    \label{eq:cycle:system}
        (B_\Lambda)_{\pa(v_1)\setminus\{v_2\}, \pa^*(C_m)}^T\cdot\Tilde{\lambda}_{\pa(v_1)\setminus\{v_2\}, v_1} = (B_\Lambda)_{v_1, \pa^*(C_m)}^T - \frac{1}{\lambda_{v_2v_1}}(B_\Lambda)_{v_2, \pa^*(C_m)}^T.
    \end{equation}    
    We know that the system 
    $$
    (B_\Lambda)_{\pa(v_1)\setminus\{v_2\}, \pa^*(C_m)}^T\cdot\Tilde{\lambda}_{\pa(v_1)\setminus\{v_2\}, v_1} = (B_\Lambda)_{v_1, \pa^*(C_m)}^T
    $$
    has always a solution given by $\lambda_{\pa(v_1)\setminus\{v_2\}}$. Hence, the system in \cref{eq:cycle:system} has a solution if and only if the system 
    \begin{equation}
        \label{eq:cycle:system:2}
        (B_\Lambda)_{\pa(v_1)\setminus\{v_2\}, \pa^*(C_m)}^T\cdot\Tilde{\lambda}_{\pa(v_1)\setminus\{v_2\}, v_1} = (B_\Lambda)_{v_2, \pa^*(C_m)}^T
    \end{equation}    
    has one.
    Using $B_\Lambda = (I-\Lambda)^{-T}$ and $\lambda_{v_1, \pa^*(C_m) = 0}$, we can write
    $$(B_\Lambda)_{v_2, \pa^*(C_m)}^T = (B_\Lambda)_{\pa(v_2)\setminus\{v_1\}, \pa^*(C_m)}^T\cdot \lambda_{\pa(v_2)\setminus\{v_1\}, v_2}.$$
    This implies that that the system in \cref{eq:cycle:system:2} has solutions for a generic choice of $\lambda_{\pa(v_2)\setminus\{v_1\}, v_2}$ if and only if the row space of $(B_\Lambda)_{\pa(v_1)\setminus\{v_2\}, \pa^*(C_m)}$ contains the row space of $(B_\Lambda)_{\pa(v_2)\setminus\{v_1\}, \pa^*(C_m)}$. That is, if $$\rank((B_\Lambda)_{\pa(v_1)\setminus\{v_2\}, \pa^*(C_m)}) = \rank((B_\Lambda)_{\pa^*(C_m), \pa^*(C_m)}).$$ From Lemma 3.5 in the main paper, one can see that this is possible if and only if the graphical condition of the theorem is satisfied.
\end{proof}

\subsection{Proofs for Section 7.2 from the Main Paper}
\label{app:subsec:estimation}
In the sequel, we denote the value of the objective function in the optimization problem of Eq.~7.1 in the main paper for a matrix \(\tilde{\Lambda} \in \mathbb{R}^{\GG_D}\) by \(\mathrm{Obj}(\tilde{\Lambda})\), and its finite-sample counterpart by \(\widehat{\mathrm{Obj}}_n(\tilde{\Lambda})\).

\begin{proof}[Proof for Lemma 7.1 from the main paper]
By definition of the map $\Phi_\GG$ we have $(I-\Lambda)^T\cdot X=\varepsilon$, that implies $\mathrm{Obj}(\Lambda)=0$. Hence, $\tilde{\Lambda}$ minimizes Eq.~7.1 from the main paper if and only if $\mathrm{Obj}(\Tilde{\Lambda})=0$, that is if and only if $$\tilde{\varepsilon}=(I-\tilde{\Lambda})\cdot X=(I-\tilde{\Lambda})B_{\Lambda}\cdot\varepsilon=A\cdot\varepsilon\in\mathcal{M}(\GG_B).$$
We know from Lemma 3.3 from the main paper that this is the case if and only if $\tilde{\Lambda}$ satisfies Eq.~3.3 from the main paper.
\end{proof}

\begin{proof}[Proof for theorem 7.2 from the main paper]
Let
\[
\mathrm{H}(R^{\Lambda}_u, R^{\Lambda}_v) := \mathrm{HSIC} \big([(I - \Lambda) \cdot X]_u, [(I - \Lambda) \cdot X]_v; \mathrm{k}_u, \mathrm{k}_v \big),
\]
and let $\mathrm{\hat{H}}_n(R^{\Lambda}_u, R^{\Lambda}_v)$ be its finite-sample counterpart.
From \citet[Cor.~15]{mooij:2016}, we know that for every $\Lambda \in \Theta$, the estimator $\hat{\mathrm{H}}_n(R^{\Lambda}_u, R^{\Lambda}_v)$ converges in probability to $\mathrm{H}(R^{\Lambda}_u, R^{\Lambda}_v)$. Moreover, for every $\Tilde{\Lambda}^1, \Tilde{\Lambda}^2 \in \R^{\GG_D}$ and $u, v \in \GG$, we have
\begin{equation}
\begin{aligned}
\label{eq:hsic:ineq}
&\hat{\mathrm{H}}_n(R^{\tilde{\Lambda}^1}_u, R^{\tilde{\Lambda}^1}_v) - \hat{\mathrm{H}}_n(R^{\tilde{\Lambda}^2}_u, R^{\tilde{\Lambda}^2}_v) \\
= & \hat{\mathrm{H}}_n(R^{\tilde{\Lambda}^1}_u, R^{\tilde{\Lambda}^1}_v) - 
\hat{\mathrm{H}}_n(R^{\tilde{\Lambda}^2}_u, R^{\tilde{\Lambda}^1}_v) + 
\hat{\mathrm{H}}_n(R^{\tilde{\Lambda}^2}_u, R^{\tilde{\Lambda}^1}_v)-
\hat{\mathrm{H}}_n(R^{\tilde{\Lambda}^2}_u, R^{\tilde{\Lambda}^2}_v) \\
\leq & \frac{32\lambda C}{\sqrt{n}}\left(||R^{\tilde{\Lambda}^1}_u - R^{\tilde{\Lambda}^2}_u|| + ||R^{\tilde{\Lambda}^1}_v - R^{\tilde{\Lambda}^2}_v||\right) \\
\leq & \frac{32\lambda C}{\sqrt{n}}\left(||X_{\pa(u)}||||\tilde{\Lambda}^1_{u, \pa(u)} - \tilde{\Lambda}^2_{u, \pa(u)}|| + ||X_{\pa(v)}||||\tilde{\Lambda}^1_{v, \pa(v)} - \tilde{\Lambda}^2_{v, \pa(v)}||\right)\\
\leq &
64\lambda C M||\tilde{\Lambda}^1-\tilde{\Lambda}^2||,
\end{aligned}
\end{equation}
where $\lambda$ and $C$ are the Lipschitz constant and upper bound for $k_u, k_v$, respectively, and $M$ is an upper bound for $||X||$, which exists since we assumed $\varepsilon$ to have compact support. In \eqref{eq:hsic:ineq}, we used \citet[Lemma 16]{mooij:2016}, the Cauchy–Schwarz inequality, and the fact that the support of $\varepsilon$ is compact.

From \citet[Prop.~2]{pfister:2018}, it follows that
\begin{equation}
    \begin{aligned}
        \mathrm{H}_{u,v}(\tilde{\Lambda}) :=& \mathrm{H}(R^{{\tilde{\Lambda}}}_u, R^{\tilde{\Lambda}}_v) \\
        =& \mathbb{E}\left[\mathrm{k}_u\left(R^{\tilde{\Lambda}}_u[1], R^{\tilde{\Lambda}}_u[2]\right)\mathrm{k}_v\left(R^{\tilde{\Lambda}}_v[1], R^{\tilde{\Lambda}}_v[2]\right)\right] \\
        +& \mathbb{E}\left[\mathrm{k}_u\left(R^{\tilde{\Lambda}}_u[1], R^{\tilde{\Lambda}}_u[2]\right)\right]\mathbb{E}\left[\mathrm{k}_v\left(R^{\tilde{\Lambda}}_v[1], R^{\tilde{\Lambda}}_v[2]\right)\right] \\
         -&2\mathbb{E}\left[\mathrm{k}_u\left(R^{\tilde{\Lambda}}_u[1], R^{\tilde{\Lambda}}_u[2]\right)\mathrm{k}_v\left(R^{\tilde{\Lambda}}_v[1], R^{\tilde{\Lambda}}_v[3]\right)\right],
    \end{aligned}
\end{equation}
where $R^{\tilde{\Lambda}}_u[1], R^{\tilde{\Lambda}}_u[2]$ are i.i.d. copies of $R^{\tilde{\Lambda}}_u$, and $R^{\tilde{\Lambda}}_v[1], R^{\tilde{\Lambda}}_v[2], R^{\tilde{\Lambda}}_v[3]$ are i.i.d. copies of $R^{\tilde{\Lambda}}_v$. Since the mappings $\Lambda \mapsto R^{\tilde{\Lambda}}_{u/v}$ and the kernels $\mathrm{k}_{u/v}$ are continuously differentiable, and all kernels have bounded derivatives, the dominated convergence theorem implies that $\mathrm{H}_{u,v}(\Lambda)$ is continuously differentiable, and hence Lipschitz on $\Theta$.

These results allow us to apply \citet[Cor.~2.2]{newey:1991}, which implies that for all $\epsilon > 0$,
\[
\lim_{n\to\infty}\mathbb{P} \left(\max_{\tilde{\Lambda} \in \Theta} \left| \hat{\mathrm{H}}_{n,u,v}(\tilde{\Lambda}) - \mathrm{H}_{u,v}(\tilde{\Lambda}) \right| > \epsilon \right) = 0.
\]
Since \(\widehat{\mathrm{Obj}}_n(\tilde{\Lambda})\) can be expressed as a finite sum of terms \(\hat{\mathrm{H}}_{n,u,v}\), uniform convergence holds:
\begin{equation}
\label{eq:conv:sum}
    \lim_{n \to \infty} \mathbb{P} \left( \max_{\tilde{\Lambda} \in \Theta} \left| \widehat{\mathrm{Obj}}_n(\tilde{\Lambda}) - \mathrm{Obj}(\tilde{\Lambda}) \right| > \epsilon \right) = 0.
\end{equation}

Define, for any \(\epsilon > 0\),
\[
\mathbb{B}_\epsilon(\mathrm{Proj}) := \{\Lambda_1 \in \mathbb{R}^\GG : \exists \Lambda_0 \in \mathrm{Proj} \text{ such that } \|\Lambda_1 - \Lambda_0\| < \epsilon \}.
\]
Set 
\[
\eta_\epsilon := \min_{\Lambda_1 \in \Theta \setminus \mathbb{B}_\epsilon(\mathrm{Proj})} \mathrm{Obj}(\Lambda_1).
\]
By compactness of \(\Theta\) and the fact that the kernels used are characteristic, \(\eta_\epsilon\) is strictly positive, well defined, and decreases to zero as \(\epsilon \to 0\).

We then have
\[
\begin{aligned}
\mathbb{P}&\Big(\min_{\Lambda_0 \in \Theta \cap \mathrm{Proj}} \|\tilde{\Lambda}_n - \Lambda_0\| > \epsilon \Big)
\leq \mathbb{P}\Big(\max_{\Lambda_1 \in \Theta \setminus \mathbb{B}_\epsilon(\mathrm{Proj})} (\widehat{\mathrm{Obj}}_n(\Lambda) - \widehat{\mathrm{Obj}}_n(\Lambda_1)) \geq 0\Big) \\
\leq \mathbb{P}&\Big(\max_{\Lambda_1 \in \Theta \setminus \mathbb{B}_\epsilon(\mathrm{Proj})} |\widehat{\mathrm{Obj}}_n(\Lambda) - \widehat{\mathrm{Obj}}_n(\Lambda_1) + \mathrm{Obj}(\Lambda_1) - \mathrm{Obj}(\Lambda)| \geq \eta_\epsilon \Big) \\
\leq \mathbb{P}&\Big(\max_{\Lambda_1 \in \Theta \setminus \mathbb{B}_\epsilon(\mathrm{Proj})} |\widehat{\mathrm{Obj}}_n(\Lambda_1) - \mathrm{Obj}(\Lambda_1)| + |\widehat{\mathrm{Obj}}_n(\Lambda) - \mathrm{Obj}(\Lambda)| \geq \eta_\epsilon \Big) \\
\leq \mathbb{P}&\Big(\max_{\Lambda_1 \in \Theta \setminus \mathbb{B}_\epsilon(\mathrm{Proj})} |\widehat{\mathrm{Obj}}_n(\Lambda_1) - \mathrm{Obj}(\Lambda_1)| \geq \tfrac{\eta_\epsilon}{2} \Big) + \mathbb{P}\Big(|\widehat{\mathrm{Obj}}_n(\Lambda) - \mathrm{Obj}(\Lambda)| \geq \tfrac{\eta_\epsilon}{2}\Big),
\end{aligned}
\]
where the first summand converges to zero due to uniform convergence \eqref{eq:conv:sum}, and the second converges to zero by consistency of HSIC.

Note that the inequalities use:
\begin{enumerate}
\item[(i)] The first inequality uses the definition of \(\tilde{\Lambda}_n\) as the empirical minimizer. Indeed, if the distance to \(\Theta \cap \mathrm{Proj}\) is larger than \(\epsilon\), it must mean that there is a point outside \(\mathbb{B}_\epsilon(\mathrm{Proj})\) for which the value of the objective function is smaller than \(\widehat{\mathrm{Obj}}_n(\Lambda)\),
\item[(ii)] the definition of \(\eta_\epsilon\) as the minimal population risk away from the \(\epsilon\)-neighborhood of \(\mathrm{Proj}\),
\item[(iii)] and the union bound to separate probabilities.
\end{enumerate}
This completes the argument.

\end{proof}
\section{Comparison to Identification by OICA}
\label{app:oica}
\citet{tramontano:2024:icml} study lvLiNGAM models and provide a graphical criterion for causal effect identification that is related to the one we propose in Theorem 3.5 of the main paper. Here, we provide a complete comparison of the two approaches. We start by recalling the definition of lvLiNGAMs.

Consider a DAG $\GG_\LL = (V_\LL = V\dot\cup L, E)$, where $V$ is the set of observed variables, $L$ is the set of latent variables, and all nodes in $L$ are source nodes. The lvLiNGAM model associated to $\GG_\LL$ is the set of distributions of random vectors 
$X_V = B\cdot\varepsilon_{V_\LL}$,
where $B = (B_{\Lambda_{\LL}})_{V, V_\LL}$ for a matrix $\Lambda_\LL\in\R^{\GG_\LL}$ and $\varepsilon_{V_\LL}\in\MM(|V_\LL|)$. We denote this set of distributions $\MM_{lv}(\GG_\LL)$.

Given a latent variable graph $\GG_\LL$, we can always consider its latent projection graph $\GG = (V, D, B)$, which will be an ADMG with vertex set $V$, where $D$ is the restriction of $E$ to the observed nodes, and $B=\{(u,v)\:: \an_\LL(u)\cap\an_\LL(v)\cap L\neq\emptyset \}$. By construction, we have $\MM_{lv}(\GG_\LL)\subset\MM(\GG)$; hence, one would expect that assuming the observed distribution belongs to $\MM_{lv}(\GG_\LL)$ allows one to identify more edges. The next proposition proves that this intuition is indeed correct. For the convenience of the reader, we report here a version of \citet[Thm.~3.9]{tramontano:2024:icml} translated into our notation.

\begin{theorem}{\cite[Thm.~3.9]{tramontano:2024:icml}}
\label{thm:dce:known:graph}
{
    Let $\GG_\LL = (V_\LL = V\dot\cup L, E)$ be a latent graph DAG. Consider any two observed variables $u$ and $v$. The direct causal effect of $u$ on $v$ is \emph{generically} identifiable in $\MM_{lv}(\GG_\LL)$ if and only if there are no pairs $(w,l)\in V\times L$ satisfying all of the following conditions:
        \begin{align}
            \label{eq:dce:cond}
            &\de(w)\cap V=\de(l)\cap V,\\
            \label{eq:dce:cond:1}
            &v\in\ch(l),\\
            \label{eq:dce:cond:2}
            &w\in\ch(u)\cup\{u\},\\                       
            \label{eq:dce:cond:3}
            &\ch(l)\setminus\ch(w_1)=\emptyset,\qquad\forall w_1\in\pa(w)\cup\{w\}.
        \end{align}}
\end{theorem}
\begin{lemma}
{Let $\GG_\LL = (V_\LL = V\dot\cup L, E)$ be a latent graph DAG, and let $\GG$ be its latent projection graph. If the causal effect from $u$ to $v$ is generically identifiable in $\MM(\GG)$, then it is generically identifiable in $\MM_{lv}(\GG_\LL)$. Moreover, there are latent variable DAGs $\GG_{\LL}$ for which some causal effect is identifiable in $\MM_{lv}(\GG_\LL)$ while the corresponding causal effect is not identifiable in $\MM(\GG)$.
}
\end{lemma}
\begin{proof}
We first prove that if the causal effect from $u$ to $v$ is not identifiable in $\MM_{lv}(\GG_\LL)$ then it is not identifiable in $\MM(\GG)$. According to Theorem 3.5 from the main paper, we need to show that if the conditions in \cref{eq:dce:cond,eq:dce:cond:1,eq:dce:cond:2,eq:dce:cond:3} are satisfied for a pair $(w,l)\in V\times L$, then  $r^v_{\pa(v)\setminus\{u\}} = r^v_{\pa(v)}$. In particular, we are going to show that to any $\Pi\in\tilde{\mathcal{P}}(I, P)$ with $I\subseteq R_v$, and $u\in P\subseteq \pa(v)$ we can associate a $\hat{\Pi}\in\tilde{\mathcal{P}}(I, P^u),$ where $P^u\subseteq \pa(v)\setminus\{u\}$.
Using \cref{eq:dce:cond:1,eq:dce:cond:3} we have that $\pa(w)\subseteq\pa(v)$ and $w\notin R_v$. This implies that if $w\in P$, then  there is $z\in\pa(w)$ such that $z\to w\in\Pi$, and $\hat{\Pi} = \Pi\setminus\{z\to w\}\in\tilde{\PP}(I, P\setminus\{w\}\cup\{z\})$. This concludes the proof in the case $u = w$. If instead $w\in\ch(u)$, one can construct $\Pi' = \hat{\Pi}\cup\{u\to w\}\in\tilde{\PP}(I, P\setminus\{u\}\cup\{w, z\})$, which concludes the first part of the proof.
\begin{figure}[ht]
        \centering
        \begin{tikzpicture}[scale = 0.8]
        \begin{scope}
        \node at (-4,0) {${\GG_\LL}:$};
 
         \node[draw, circle, inner sep=4pt] (I) at (-3,0) {$v_1$};
        
        \node[draw, circle, inner sep=4pt] (T1) at (1,1.5) {$v_2$};
        \node[draw, circle, inner sep=4pt] (T2) at (1,-1.5) {$v_3$};

        \node[draw, circle, inner sep=4pt] (Y) at (5,0) {$v_4$};
        
        \node[draw, circle, fill=gray!40, inner sep=4pt] (H0) at (-1.5,2) {$l_0$};
        \node[draw, circle, fill=gray!40, inner sep=4pt] (H1) at (3.5,2) {$l_1$};
        \node[draw, circle, fill=gray!40, inner sep=4pt] (H2) at (3.5,-2) {$l_2$};

        \draw[->] (I) to (T1);
        \draw[->] (I) to (T2);
        
        \draw[->] (T1) to (Y);
        \draw[->] (T2) to (Y);

        \draw[black, ->] (H0) to (T1);
        \draw[black, ->] (H0) to (I);
        
        \draw[black, ->] (H1) to (T1);
        \draw[black, ->] (H1) to (Y);

        \draw[black, ->] (H2) to (T2);
        \draw[black, ->] (H2) to (Y);
        \end{scope}
        
        \begin{scope}[xshift = 10cm]
        \node at (-2,0) {${\GG}:$};
         \node[draw, circle, inner sep=4pt] (I) at (-1,0) {$v_1$};
        
        \node[draw, circle, inner sep=4pt] (T1) at (1,1) {$v_2$};
        \node[draw, circle, inner sep=4pt] (T2) at (1,-1) {$v_3$};

        \node[draw, circle, inner sep=4pt] (Y) at (4,0) {$v_4$};
        
        \draw[->] (I) to (T1);
        \draw[->] (I) to (T2);
        
        \draw[->] (T1) to (Y);
        \draw[->] (T2) to (Y);
        
        \draw[dashed][<->] (T1) to[bend right=30] (I);
        \draw[dashed][<->] (T1) to[bend left=30] (Y);
        \draw[dashed][<->] (T2) to[bend right=30] (Y);
        \end{scope}
  
        \end{tikzpicture}
        \caption{An example of a latent variable graph for which all the parameters are identifiable, while some of them are not in the corresponding latent projection graph.}
        \label{fig:UIV}
    \end{figure}
For the second part of the proof, we construct an explicit example of a latent variable graph $\GG_\LL$ in which the causal effect is generically identifiable in $\MM(\GG_\LL)$, but not in its corresponding latent projection. To this end, consider the latent variable graph shown in \cref{fig:UIV}, and focus on the nodes $v_2$ and $v_4$. The only latent parent of $v_4$ is the node $l_1$, and the only observed node $w$ for which \cref{eq:dce:cond} holds is $v_2$. Hence, we must have $(w, l) = (v_2, l_1)$. However, since $v_1 \in \pa(v_2)$ and $v_4 \in \ch(l_1) \setminus \ch(v_1)$, condition \cref{eq:dce:cond:3} is not satisfied. This shows that the causal effect from $v_2$ to $v_4$ is generically identifiable in $\MM(\GG_\LL)$. The proof that the corresponding causal effect is not identifiable in the latent projection graph is given in Example 4.1 of the main paper.
\end{proof}

\newpage
\section{Details for the Experiments in the Main Paper}
\label{app:simulations}
\subsection{Data Generation}
\label{app:data:gen}
\subsubsection{Identification}
For fixed sample size $p$ and edge density $d$, the ADMGs for the experiments in Section 7.1 in the main paper are generated as follows:
\begin{enumerate}
    \item{ compute $e:= \lfloor dp(p-1) \rfloor$,}
    \item sample a random integer $e_d$ in $\{1,\dots,e\}$,
    \item let $\GG_D$ be a randomly generated DAG with $p$ nodes and $e_d$ edges,
    \item let $\GG_B$ be a randomly generated undirected graph with $p$ nodes and $e-e_d$ edges,
    \item define $\GG$ as $([p], \GG_D, \GG_B)$.
\end{enumerate}

\subsubsection{Estimation}
The data for the experiments in Section 7.2 in the main paper are generated as follows: 
\begin{enumerate}
    \item for every $v\in V$ we sample $\eta_v$ from a Laplace distribution with mean zero and standard deviation $s_v\sim \text{U}(0.2, 3)$,
    \item for every $u\xleftrightarrow{}v\in\GG_B$ we sample two independent random vectors $\eta^1_{u,v}, \eta^2_{u,v}$, again with standard deviations $s^1_{u,v}, s^2_{u,v}\sim \text{U}(0.2, 3)$,
    \item for every $v\in V$, we have $\varepsilon_v = \eta_v + \sum_{u\xleftrightarrow{}v\in\GG_B}(w^{v,1}_{uv}\eta^1_{uv} + w^{v,2}_{uv}\eta^2_{uv})$, where $w^{v,1}_{uv}, w^{v,2}_{uv}\sim \text{U}(-5,5)$,
    \item for every $u\to v\in\GG$, $\lambda_{uv}\sim \text{U}(-5,5)$, and $X_v = \sum_{u\in\pa(v)}\lambda_{uv} X_u + \varepsilon_v$.
\end{enumerate}

\subsection{Additional Experiments}
\label{app:add:exp}
\cref{fig:laplace_rbf} shows the performance of our methods when using the RBF kernel, with bandwidth computed using the median heuristic. We see that, compared to the polynomial kernel, this choice seems to suffer more from the non-convexity of the objective function. In contrast, it provides a better estimate when initialized at the true parameter value.

\begin{figure}[ht]
    \centering
    \includegraphics[width=\linewidth]{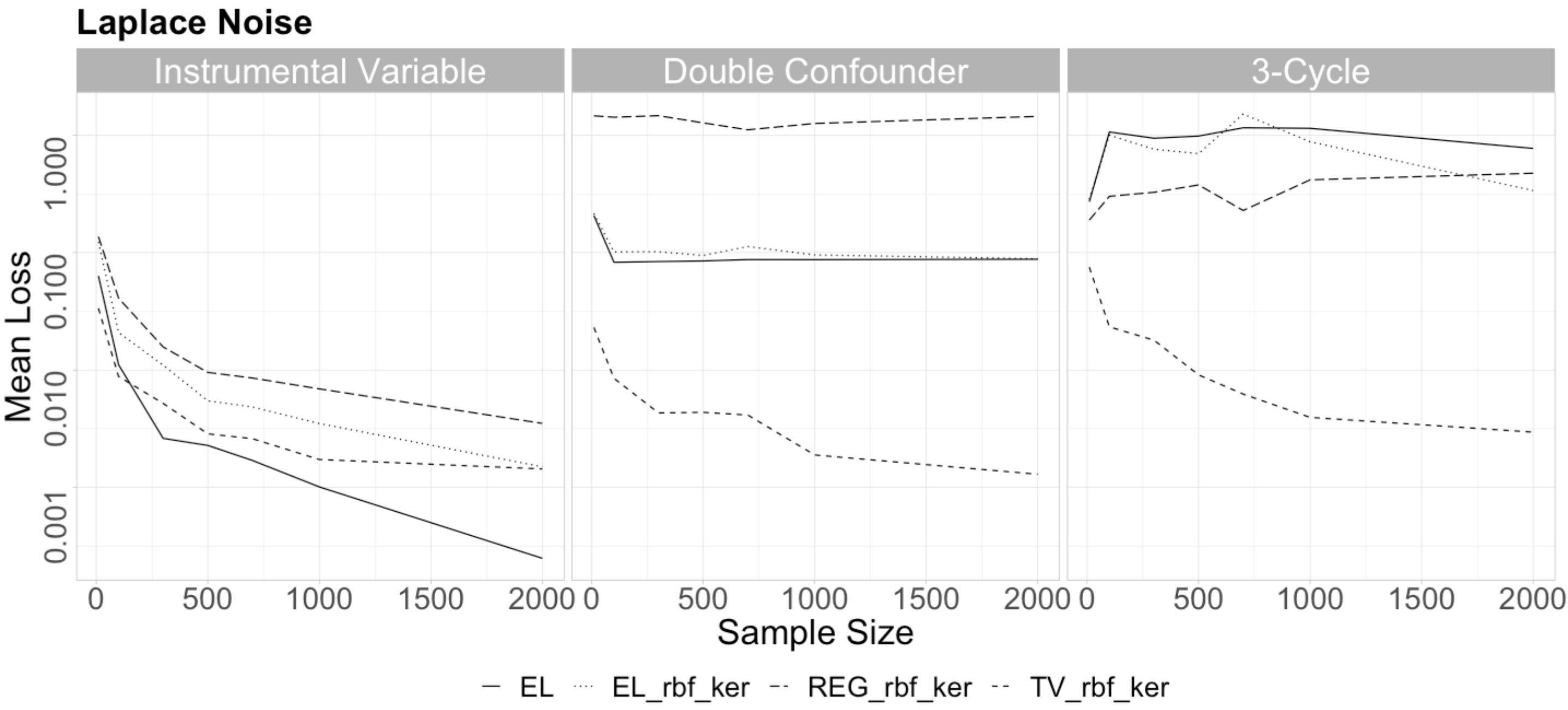}
    \caption{The performance of our proposed estimator with a Gaussian kernel, for different initial values, ``EL" stands for Empirical Likelihood, ``REG" for regression coefficient, and ``TV" for true value. We use the normalized Frobenious loss between the estimated matrix $\hat{\Lambda}$, and the true matrix $\Lambda$, i.e., $||\hat{\Lambda} - \Lambda||_F/||\Lambda||_F$, as loss function. We report the mean loss over fifty randomly sampled $\Lambda$. Notice that the $y$-axis is on a log-scale.}
    \label{fig:laplace_rbf}
\end{figure}
\cref{fig:varying_degree} shows the performance of our method as a function of the degree of the polynomial kernel. For each degree, we display a boxplot of the resulting loss values, with the sample size fixed at 500. We observe that while the algorithm remains relatively stable across degrees, higher-degree polynomials exhibit greater sensitivity to the choice of initialization. Similar to the behavior observed with the RBF kernel, more complex models (i.e., those with higher polynomial degrees) suffer more from the non-convexity of the objective function. This is reflected in the tendency of these models to achieve better performance when initialized at the true parameter value, while showing degraded performance when initialized at the regression coefficient. From a computational standpoint, lower-degree polynomials are also advantageous, as they reduce kernel evaluation complexity without significantly impacting the overall performance.
 
\begin{figure}[ht]
    \centering
    \includegraphics[scale = 0.35]{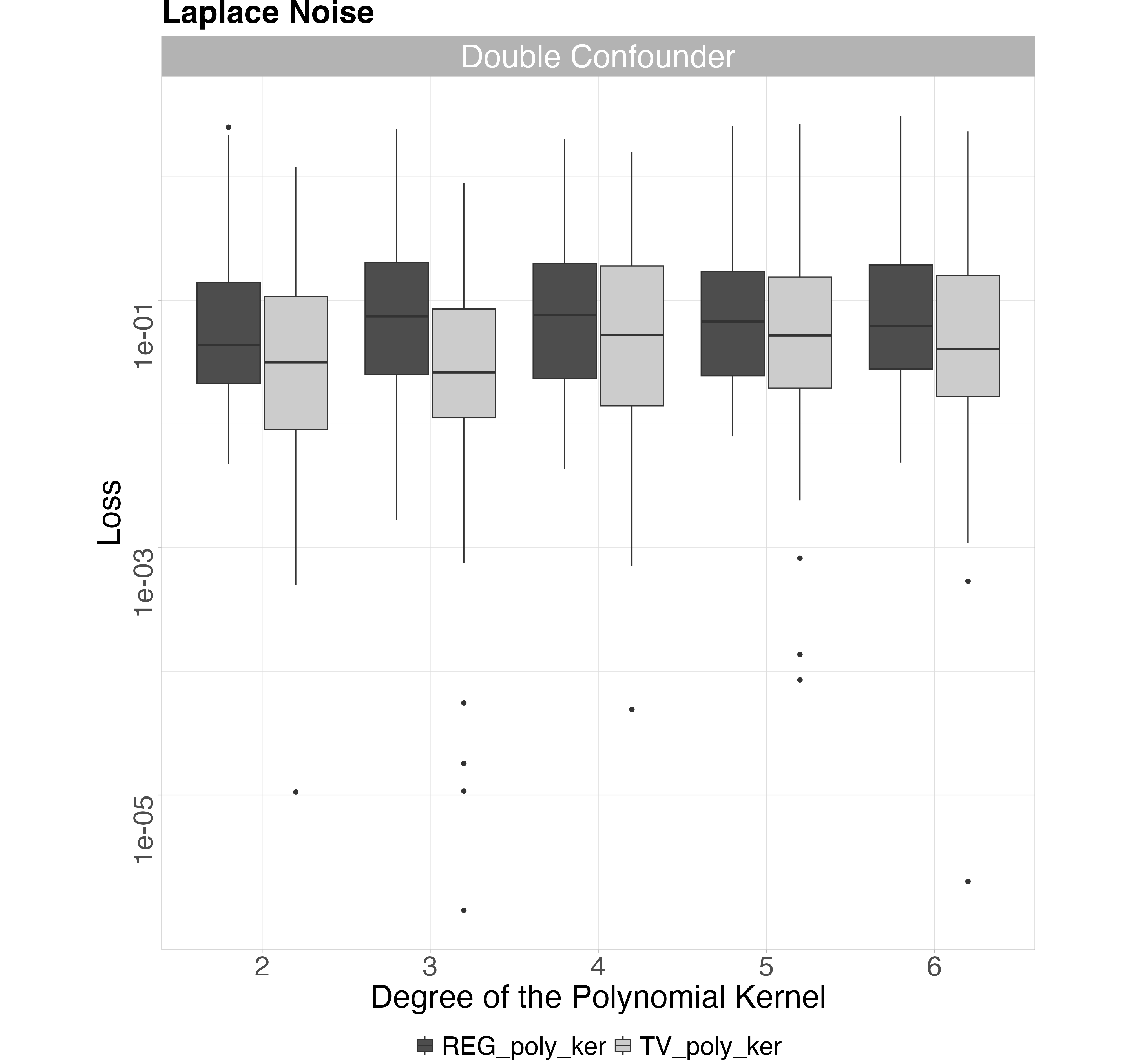}
    \caption{The performance of our proposed estimator with a varying polynomial kernel, for different initial values, ``REG" for regression coefficient, and ``TV" for true value. We use the normalized Frobenious loss between the estimated matrix $\hat{\Lambda}$, and the true matrix $\Lambda$, i.e., $||\hat{\Lambda} - \Lambda||_F/||\Lambda||_F$, as loss function. We report the boxplot of the loss over fifty randomly sampled $\Lambda$. Notice that the $y$-axis is on a log-scale. The sample size is fixed at 500.}
    \label{fig:varying_degree}
\end{figure}
\Cref{fig:one_edge_identifiable} illustrates the performance of our estimator in a case where only one of the two edges is identifiable. As expected, the estimator consistently recovers only the identifiable parameter.
\begin{figure}[ht]
    \centering
    \includegraphics[scale = 0.15]{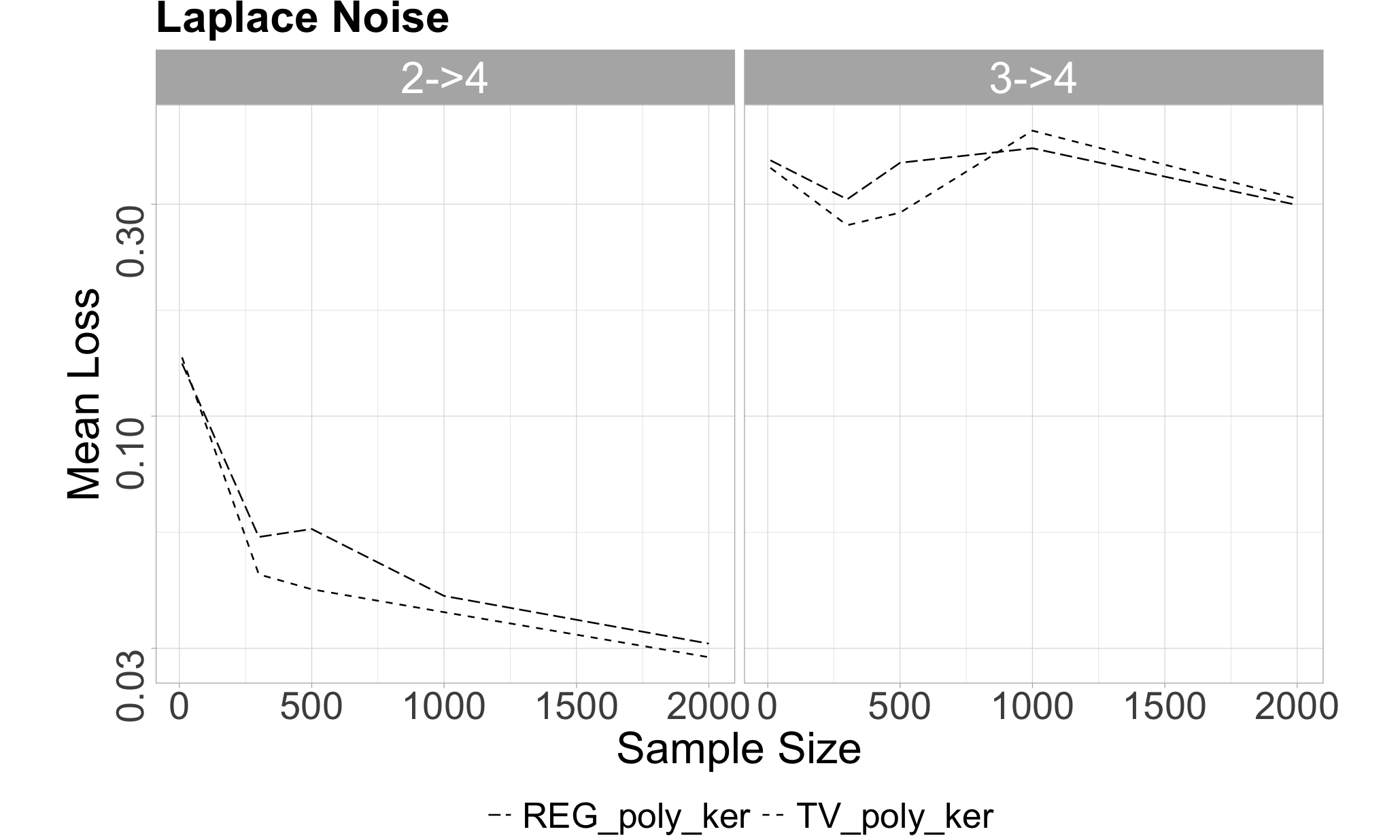}
    \includegraphics[scale = 0.15]{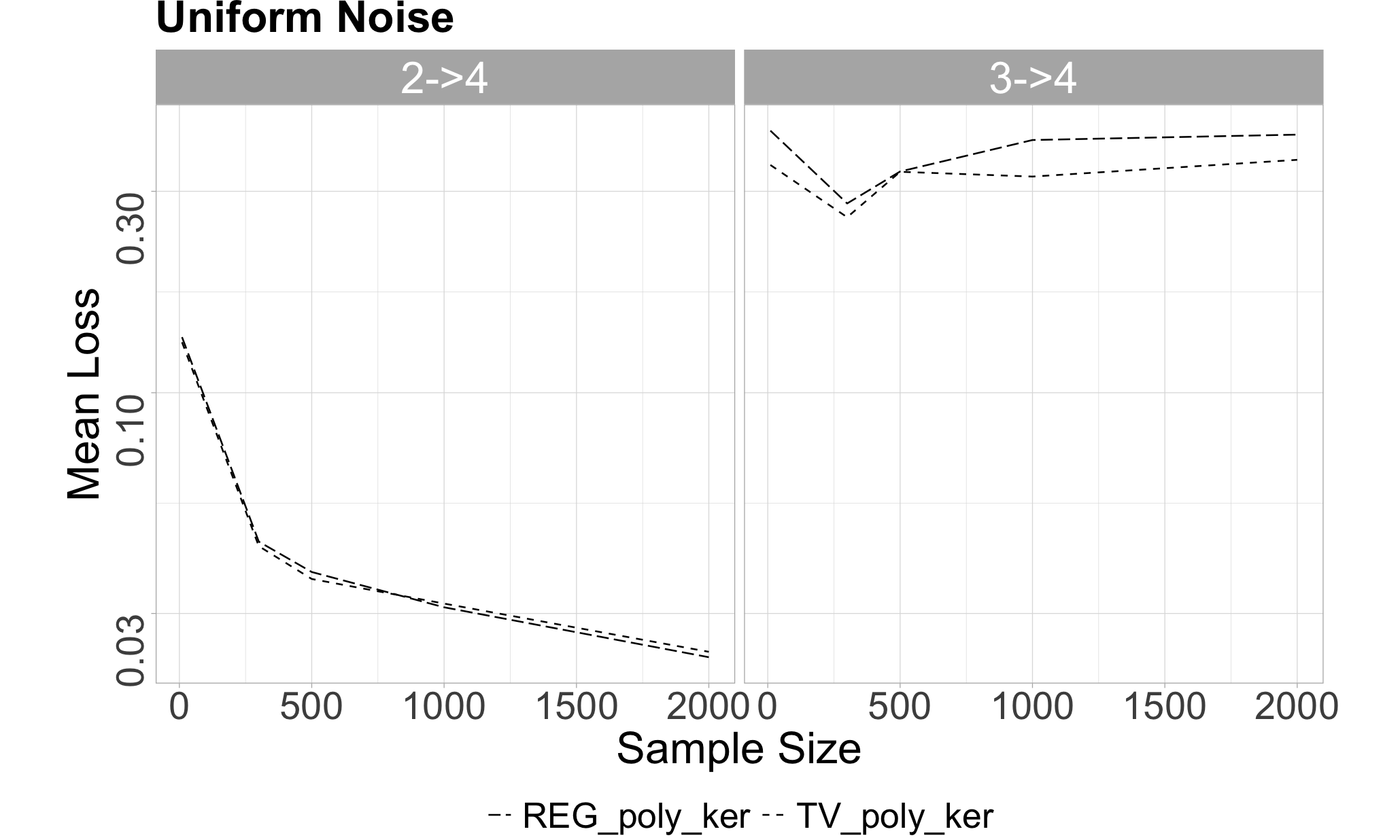}
    \caption{The performance of the proposed estimator on the graph in Fig.~4 from the main paper. We report the mean loss over fifty randomly sampled $\Lambda$. Notice that the $y$-axis is on a log-scale.}
    \label{fig:one_edge_identifiable}
\end{figure}
    
    The results for the same data-generating process, but using a uniform distribution for the error terms, are shown in \cref{fig:unif}.
    \begin{figure}[b]
    \centering
    \includegraphics[width=\linewidth]{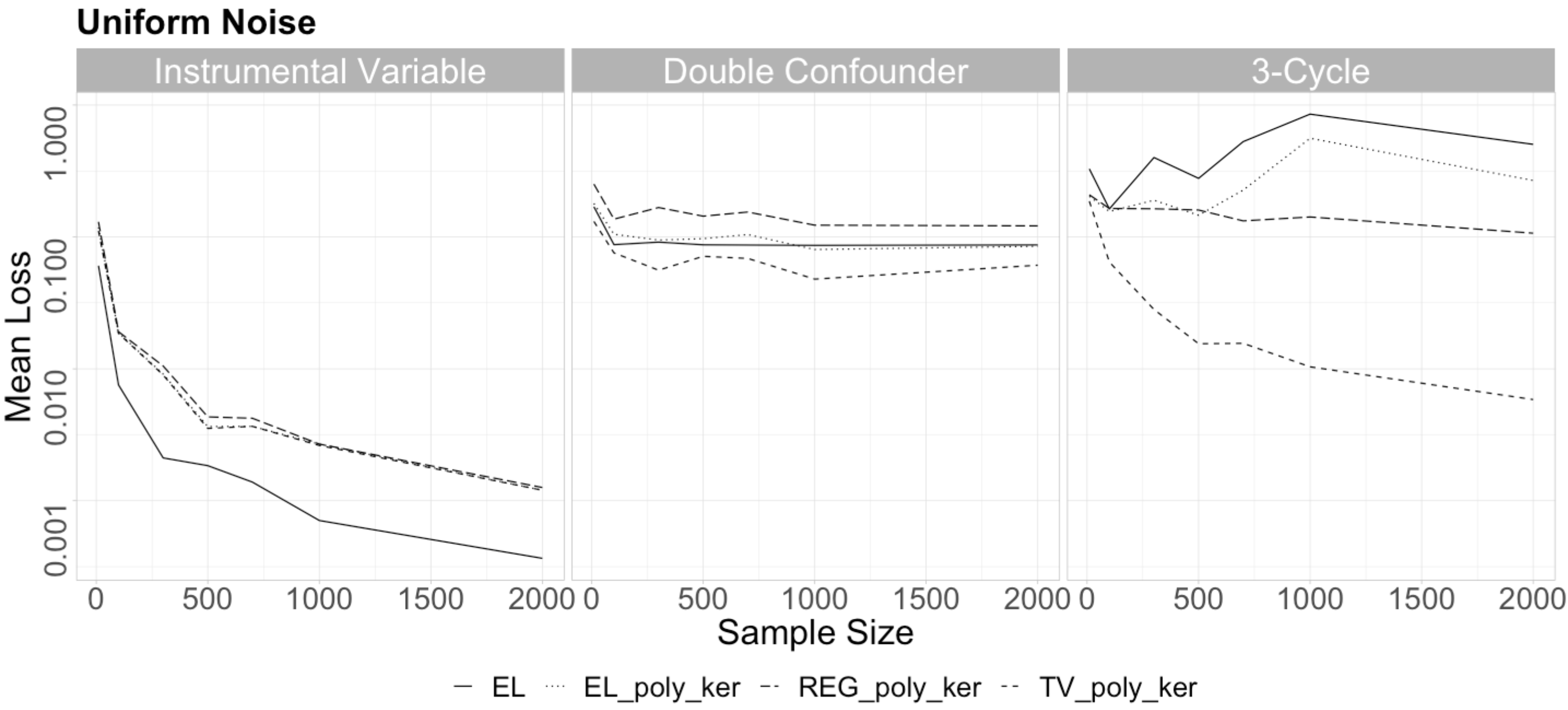}
    \includegraphics[width=\linewidth]{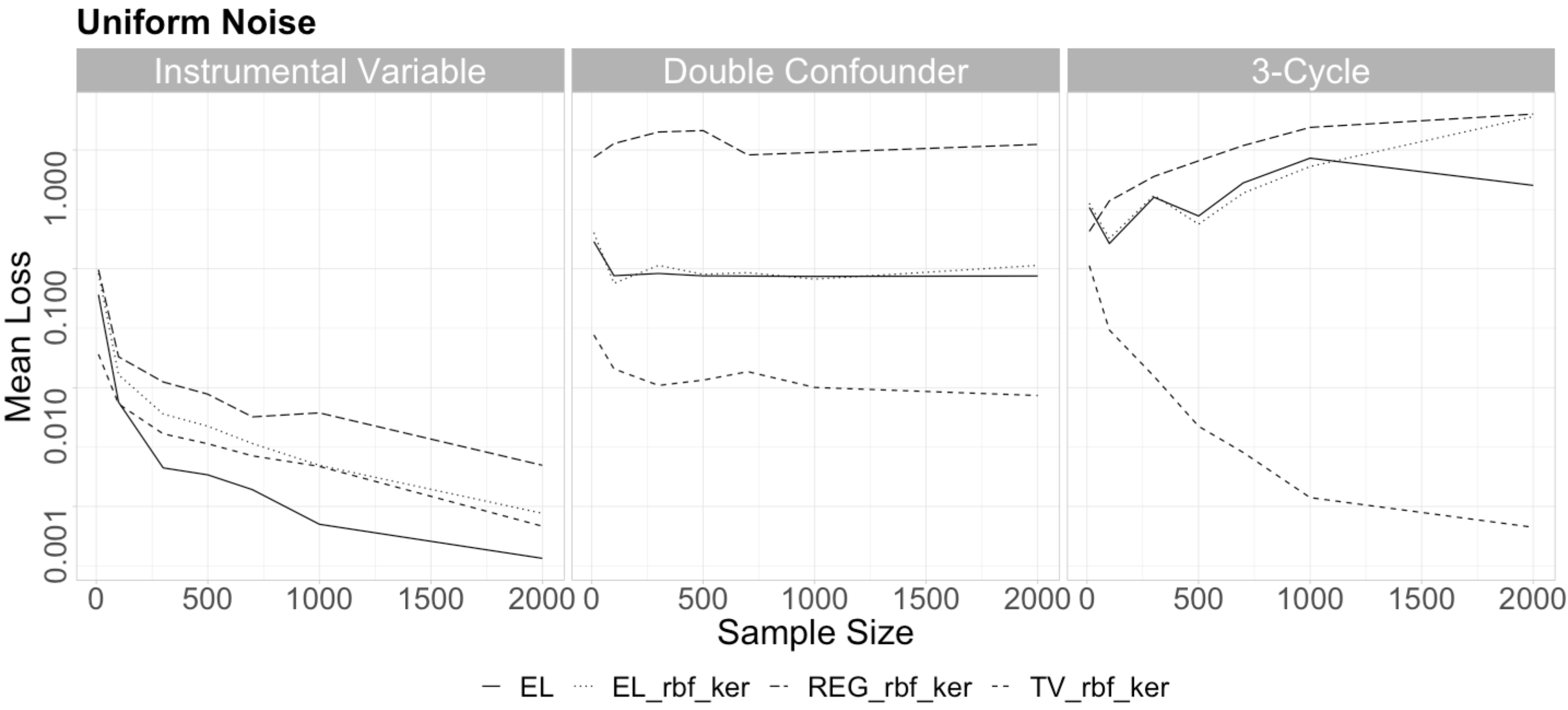}
    \caption{The performance of our proposed estimator with a polynomial kernel (top) and RBF kernel (bottom) for different initial values. ``EL" stands for Empirical Likelihood, ``REG" for regression coefficient, and ``TV" for true value. We use the normalized Frobenius loss between the estimated matrix $\hat{\Lambda}$ and the true matrix $\Lambda$, i.e., $||\hat{\Lambda} - \Lambda||_F/||\Lambda||_F$, as the loss function. We report the mean loss over fifty randomly sampled $\Lambda$. Notice that the $y$-axis is on a log-scale. The noise vector is sampled from a uniform distribution.}
    \label{fig:unif}
\end{figure}
\clearpage
\bibliographystyle{imsart-nameyear} 
\bibliography{ref}       

\end{document}